\documentclass[american,10pt,letterpaper,aps,prl,twocolumn,superscriptaddress,floatfix]{revtex4-1}
\usepackage[T1]{fontenc}
\usepackage[latin9]{inputenc}
\usepackage{bm}
\usepackage{amsmath}
\usepackage{amssymb}
\usepackage{cancel}
\usepackage{graphicx}

\makeatletter
\usepackage{braket}
\usepackage{bm}
\usepackage{amsthm}
\newtheorem{thm}{Theorem}
\newtheorem{cor}{Corollary}
\newtheorem{lem}{Lemma} 
\date{\today}

\makeatother

\usepackage{babel}
\begin{document}
\title{Optimal measurements for quantum multiparameter estimation with general
states}
\author{Jing Yang}
\email{jyang75@ur.rochester.edu}

\affiliation{Department of Physics and Astronomy, University of Rochester, Rochester,
New York 14627, USA}
\affiliation{Center for Coherence and Quantum Optics, University of Rochester,
Rochester, New York 14627, USA}
\author{Shengshi Pang}
\email{pangshengshi@gmail.com}

\affiliation{Department of Physics and Astronomy, University of Rochester, Rochester,
New York 14627, USA}
\affiliation{Center for Coherence and Quantum Optics, University of Rochester,
Rochester, New York 14627, USA}
\affiliation{Fermilab, P. O. Box 500, Batavia, Illinois 60510, USA}
\author{Yiyu Zhou}
\email{yzhou62@ur.rochester.edu}

\affiliation{The Institute of Optics, University of Rochester, Rochester, New York
14627, USA}
\author{Andrew N. Jordan}
\email{jordan@pas.rochester.edu}

\affiliation{Department of Physics and Astronomy, University of Rochester, Rochester,
New York 14627, USA}
\affiliation{Center for Coherence and Quantum Optics, University of Rochester,
Rochester, New York 14627, USA}
\affiliation{Institute for Quantum Studies, Chapman University, 1 University Drive,
Orange, CA 92866, USA}
\begin{abstract}
We generalize the approach by Braunstein and Caves {[}Phys. Rev. Lett.
72, 3439 (1994){]} to quantum multi-parameter estimation with general
states. We derive a matrix bound of the classical Fisher information
matrix due to each measurement operator. The saturation of all these
bounds results in the saturation of the matrix Helstrom Cram\'er-Rao
bound. Remarkably, the saturation of the matrix bound is equivalent
to the saturation of the scalar bound with respect to any given positive
definite weight matrix. Necessary and sufficient conditions are obtained
for the optimal measurements that give rise to the Helstrom Cram\'er-Rao
bound associated with a general quantum state. To saturate the Helstrom
bound with separable measurements or collective measurement entangling
only a small number of identical states, we find it is necessary for
the symmetric logarithmic derivatives to commute on the support of the
state. As an important application of our results, we construct several
local optimal measurements for the problem of estimating the three-dimensional
separation of two incoherent optical point sources. 
\end{abstract}
\maketitle

\section{Introduction}

Metrology \citep{giovannetti2004quantumenhanced,wiseman2009quantum,giovannetti2011advances},
the science of precision measurements, has found wide applications
in various fields of physics and engineering, including interferometry
\citep{helm2018spinorbitcoupled,degen2017quantum,haine2016meanfield,stevenson2015sagnacinterferometry,cronin2009opticsand},
atomic clocks \citep{ludlow2015optical,derevianko2011colloquium,muessel2014scalable},
optical imaging \citep{tsang2016quantum,genovese2016realapplications,dowling2015quantum,kolobov2007quantum},
and detection of gravitational waves \citep{adhikari2014gravitational}.
In classical metrology, the covariance matrix of a maximum likelihood
estimator can always asymptotically achieve the classical Cram\'er-Rao
(CR) bound proportional to the inverse of the classical Fisher information
matrix (CFIM) \citep{cramer1946mathematical,kay1993fundamentals}.
In quantum metrology, the CR bound can be further minimized over all
possible quantum measurements to yield the its quantum generalization.
However, over the years different quantum generalizations of the CR
bound have been developed motivated by the fact that the minimum variance
for all the parameters may not be achievable simultaneously in a single
measurement. The  strongest bound in the current literature is
the Holevo CR bound \citep{holevo2011probabilistic,hayashi2005asymptotic,kahn2009localasymptotic,suzuki2016explicit}.
However, this bound involves a complicated minimization over a set
of locally unbiased operators. Other well-known bounds include the
Yuen and Lax \citep{yuen1973multipleparameter} CR bound which is
in terms of the right logarithmic derivatives and the Helstrom CR
bound \citep{helstrom1976quantum2,helstrom1968theminimum}. But the most
widely used quantum CR bound by quantum physicists perhaps is the
one proposed by Helstrom due to its simplicity and intimate connections
to the geometric structure of quantum states \citep{bengtsson2017geometry}.
Therefore we focus on the saturation of the Helstrom CR bound in this
paper. In the Helstrom CR bound, the CFIM is maximized over all positive operator-valued measure (POVM)
measurements to give the quantum Fisher information matrix (QFIM).
Throughout this paper, we assume the limit of the large sample size.
Therefore the saturation of the Helstrom CR bound, becomes the search
for optimal measurements that saturates the QFIM. Braunstein and Caves
showed \citep{braunstein_statistical_1994} that for single parameter
estimation such an optimal measurement always exists. However the
QFIM in multi-parameter estimation in general may not be achievable
by any quantum measurement even in the asymptotic sense of the large
sample size \citep{szczykulska2016multiparameter}. 

On the other hand, a general theory of quantum multi-parameter estimation
is desired in many practical scenarios, including superresolution
\citep{yu2018quantum,backlund2018fundamental,zhou2018amodern,vrehavcek2017multiparameter,Rehacek-17-OL,tham2017beating,ang2017quantum,paur2016achieving,nair2016farfield,lupo2016ultimate,yang2016farfield}
spurred by the seminal work \citep{tsang2016quantum}, Hamiltonian
estimation \citep{liu2017controlenhanced,yuan2016sequential,baumgratz2016quantum,humphreys2013quantum},
parameter estimation in interferometry \citep{roccia2017entangling,pezz`e2017optimal,ragy2016compatibility,vidrighin2014jointestimation,genoni2013optimal}.
For a pure state, the saturation of the Helstrom CR bound is now fully
understood due to the works of Matsumoto \citep{matsumoto2002anew}
and Pezz\`e et al \citep{pezz`e2017optimal}. However, for a general
mixed state, the necessary and sufficient conditions for any 
POVM measurement to saturate the Helstrom
CR bound are still uncharted. In this paper, we derive the saturation
conditions by generalizing the earlier approach developed by Braunstein
and Caves \citep{braunstein_statistical_1994} for single parameter
estimation to multi-parameter estimation. For the POVM operator corresponding
to a zero probability outcome, we find that the saturation of the
Helstrom CR bound imposes a constraint which is satisfied automatically
in the case of single parameter estimation and therefore does not
appear there. It is also found that for the existence of optimal measurements
it is necessary to have the symmetric logarithmic derivatives commute
on the support of a state. Based on the saturation conditions, we
also construct several local optimal measurements in the problem of
estimating the three-dimensional separation of two monochromatic,
incoherent point sources, which is shown in Fig. \ref{fig:Setup}.
We emphasize that the saturation conditions we find may have possible
applications not only in the superresolution of optical imaging, but
also in quantum sensing \citep{degen2017quantum}.

\begin{figure}
\begin{centering}
\includegraphics[scale=0.3]{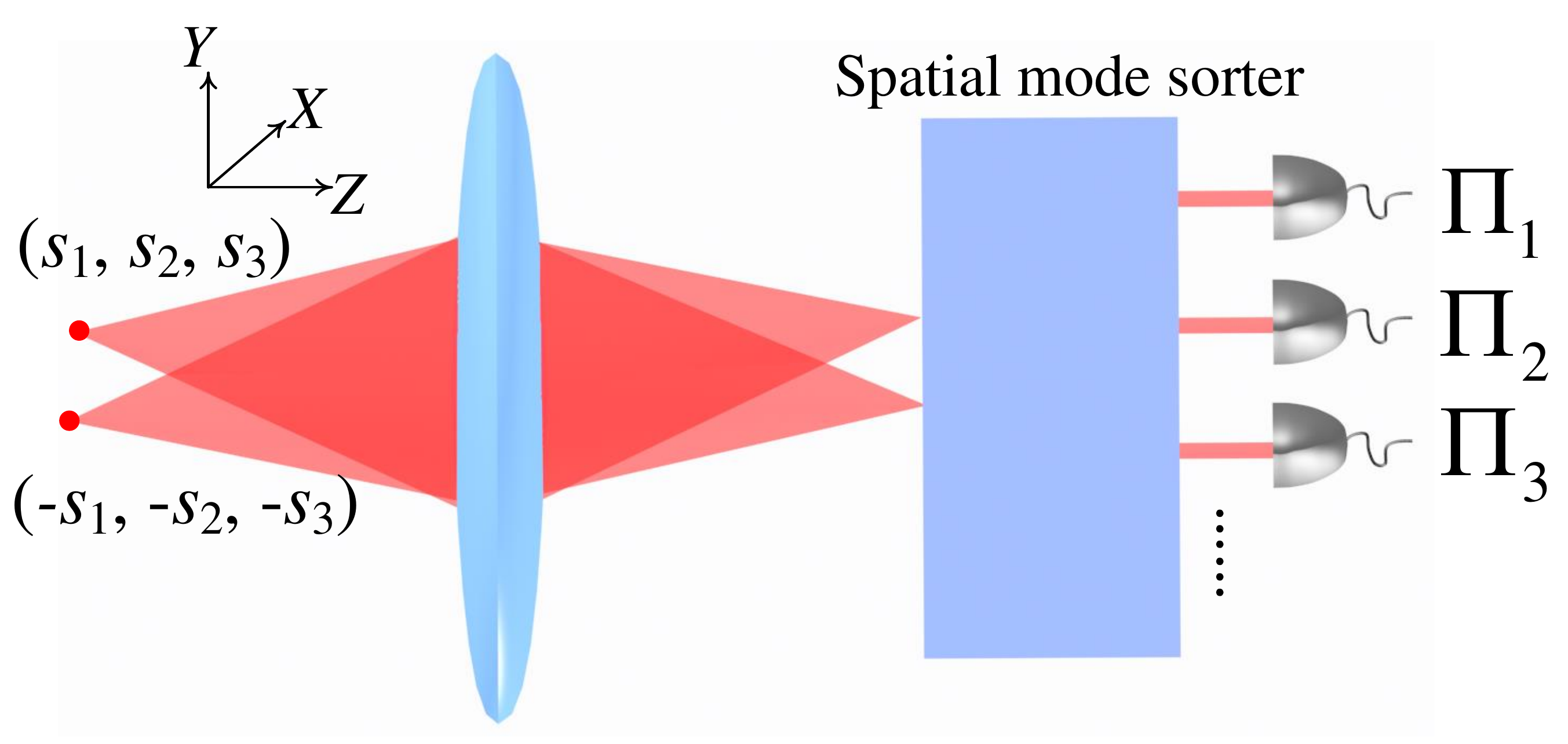}
\par\end{centering}
\caption{Schematic setup of the estimation of the three-dimensional separation
of two point light sources, whose coordinates are denoted as $\pm(s_{1},\,s_{2},\,s_{3})$.
$\Pi_{k}$ denotes a projector of a measurement, which can be implemented
by spatial mode sorters \citep{linares2017interferometric,linares2017spatial,zhou2017sorting}.
The measurement is performed at the image plane. Alternatively, the
measurement can also be performed at the pupil plane, i.e., the Fourier
transformed plane of the image plane (not shown). \label{fig:Setup}}
\end{figure}

This paper is organized as follows: In Sec.~\ref{sec:Preliminary},
we define the notations and concepts which is necessary for the subsequent
derivations. In Sec.~\ref{sec:recovering}, we recover the Helstrom
CR bound by generalizing the approach of Braunstein and Caves \citep{braunstein_statistical_1994}
and show that the saturation of the matrix Helstrom CR bound is equivalent
to the saturation of the corresponding scalar bound with respect to
any given positive definite cost matrix. In Sec.~IV, we derive the
necessary and sufficient conditions for the optimal measurements that
saturate the matrix Helstrom CR bound. In Sec.~\ref{sec:superresolution},
we apply our saturation conditions to the problem of three-dimensional
imaging of two optical incoherent point sources. We summarize our
findings and discuss several open questions in Sec.~\ref{sec:Conclusion}.

\section{Preliminary notations and definitions\label{sec:Preliminary}}

Before starting our derivations, some notations and definitions are
in order for later use: (a) A general probe state is described by
the density operator 
\begin{equation}
\rho_{\bm{\lambda}}=\sum_{n}p_{n\bm{\lambda}}\ket{\psi_{n\bm{\lambda}}}\bra{\psi_{n\bm{\lambda}}},
\end{equation}
where $p_{n\bm{\lambda}}$'s are \textit{strictly} positive and $\ket{\psi_{n\bm{\lambda}}}$'s
are orthonormal and do not vanish globally for all $\bm{\lambda}$.
We denote the \textit{kernel (null space)} of $\rho_{\bm{\lambda}}$
at some specific value $\bm{\lambda}_{0}$ as 
\begin{equation}
\text{ker}(\rho_{\bm{\lambda}_{0}})\equiv\text{span}\{\ket{\psi}:\braket{\psi_{n\bm{\lambda}_{0}}\big|\psi}=0,\,\forall n\},
\end{equation}
and the \textit{support} of $\rho_{\bm{\lambda}}$ at $\bm{\lambda}_{0}$
as 
\begin{equation}
\text{supp}(\rho_{\bm{\lambda}_{0}})\equiv\text{span}\{\ket{\psi_{n\bm{\lambda}_{0}}}\text{'s}\}.
\end{equation}
For a vector $\ket{\psi}$, its projection on $\text{ker}(\rho_{\bm{\lambda}})$
is denoted as $\ket{\psi^{0}}$ and projection on $\text{supp}(\rho_{\bm{\lambda}})$
is denoted as $\ket{\psi^{\perp}}$. According to linear algebra there
is a unique decomposition $\ket{\psi}=\ket{\psi^{0}}+\ket{\psi^{\perp}}$.
(b) We use a short hand notation $\partial_{i}$ as the derivative
with respect to the estimation parameter $\lambda_{i}$, for example
$\partial_{i}\rho_{\bm{\lambda}}\equiv\partial\rho_{\bm{\lambda}}/\partial\lambda_{i}$.
In addition the projections of $\ket{\partial_{i}\psi_{\bm{\lambda}}}$
on the kernel and support of $\rho_{\bm{\lambda}}$ are denoted as
$\ket{\partial_{i}^{0}\psi_{\bm{\lambda}}}$ and $\ket{\partial_{i}^{\perp}\psi_{\bm{\lambda}}}$
respectively, where 
\begin{equation}
\ket{\partial_{i}^{0}\psi_{\bm{\lambda}}}\equiv\ket{\partial_{i}\psi_{\bm{\lambda}}}-\ket{\partial_{i}^{\perp}\psi_{\bm{\lambda}}},
\end{equation}
 
\begin{equation}
\ket{\partial_{i}^{\perp}\psi_{\bm{\lambda}}}\equiv\sum_{n}\ket{\psi_{n\bm{\lambda}}}\braket{\psi_{n\bm{\lambda}}\big|\partial_{i}\psi_{\bm{\lambda}}}.
\end{equation}
(c) The POVM operator is denoted as $\Pi_{k}$ with spectral decomposition
\begin{equation}
\Pi_{k}\equiv\sum_{\alpha}q_{k\alpha}\ket{\pi_{k\alpha}}\bra{\pi_{k\alpha}},
\end{equation}
where $q_{k\alpha}$'s are strictly positive and $\ket{\pi_{k\alpha}}$'s
are orthonormal. If $\text{Tr}(\rho_{\bm{\lambda}}\Pi_{k})=0$ then
$\Pi_{k}$ is called a \textit{null }(POVM) operator otherwise it
is called a\textit{ regular} (POVM) operator. A basis vector $\ket{\pi_{k\alpha}}$
is \textit{null} if $\braket{\psi_{n\bm{\lambda}}\big|\pi_{k\alpha}}=0,\,\forall n$
otherwise it is \textit{regular}. We emphasis that null operator will
make the CFIM elements ill-defined and therefore some regularization
is required when calculating their contributions to the CFIM.

\section{Recovering the Helstrom CR bound\label{sec:recovering}}

The CFIM quantifies the sensitivity of a probability distribution
to a small change in $\bm{\lambda}$ \citep{bengtsson2017geometry}.
Its matrix element is defined as \citep{cramer1946mathematical,kay1993fundamentals}
\begin{equation}
F_{ij}=\sum_{k}\mathcal{F}_{ij}^{k},
\end{equation}
where 
\begin{equation}
\mathcal{F}_{ij}^{k}\equiv\partial_{i}\text{Tr}(\rho_{\bm{\lambda}}\Pi_{k})\partial_{j}\text{Tr}(\rho_{\bm{\lambda}}\Pi_{k})/\text{Tr}(\rho_{\bm{\lambda}}\Pi_{k}).
\end{equation}
Note that null operators contribute to the CFIM elements terms of
the type $0/0$, which should be understood in the sense of the multivariate
limit. Due to this observation, it is natural to discuss the CFIM
element separately for null and regular operators. For both regular
and null operators, we prove the following inequality in subsecs \ref{subsec:regular}
and \ref{subsec:null}
\begin{equation}
\sum_{ij}u_{i}\mathcal{F}_{ij}^{k}u_{j}\leq\sum_{ij}u_{i}u_{j}\mathcal{I}_{ij}^{k},\label{eq:uFku}
\end{equation}
where $\bm{u}$ is an arbitrary, real, and nonzero vector \footnote{For a complex vector ${z_{i}=u_{i}+\text{i}v_{i}}$ and a real and
symmetric matrix ${A_{ij}}$, ${\sum_{ij}z_{i}^{*}A_{ij}z_{j}=\sum_{ij}u_{i}A_{ij}u_{j}+\sum_{ij}v_{i}A_{ij}v_{j}}$.
Thus we see that $\sum_{ij}z_{i}^{*}A_{ij}z_{j}>0$ for any real and
nonzero $z_{i}$ is equivalent as ${\sum_{ij}u_{i}A_{ij}u_{j}>0}$
for any real and nonzero vector $u_{i}$. To show the matrix ${\mathcal{I}_{ij}^{k}-\mathcal{F}_{ij}^{k}}$
is positive definite, it is sufficient to consider real and nonzero
vectors in Eq.~(\ref{eq:uFku}).}, $L_{i\bm{\lambda}}$ is the Symmetric Logarithmic Derivative (SLD)
with respect to parameter $\lambda_{i}$ defined as \citep{paris_quantum_2009,holevo2011probabilistic,helstrom1976quantum2}
\begin{equation}
[L_{i\bm{\lambda}}\rho_{\bm{\lambda}}+\rho_{\bm{\lambda}}L_{i\bm{\lambda}}]/2=\partial_{i}\rho_{\bm{\lambda}},\label{eq:SLD}
\end{equation}
and 
\begin{equation}
\mathcal{I}_{ij}^{k}\equiv\text{Re}[\text{Tr}(\rho_{\bm{\lambda}}L_{i\bm{\lambda}}\Pi_{k}L_{j\bm{\lambda}})]
\end{equation}
is the QFIM element corresponding to a regular or null POVM operator
$\Pi_{k}$. Note that Eq.~(\ref{eq:uFku}) for null operators is
not discussed in Ref.~\citep{braunstein_statistical_1994} since
it is automatically saturated in the case of single parameter estimation,
as we will see in the next section. It is readily checked that summation
over $k$ in Eq.~(\ref{eq:uFku}) yields the Helstrom CR bound \citep{helstrom1976quantum2,holevo2011probabilistic,paris_quantum_2009}
\begin{equation}
\sum_{ij}u_{i}F_{ij}u_{j}\leq\sum_{ij}u_{i}I_{ij}u_{j},\label{eq:HelstromCR}
\end{equation}
where
\begin{equation}
I_{ij}\equiv\text{Re}[\text{Tr}(\rho_{\bm{\lambda}}L_{i\bm{\lambda}}L_{j\bm{\lambda}})].
\end{equation}

\subsection{\label{subsec:regular}Proof of Eq. (\ref{eq:uFku}) for regular
POVM operators}

\begin{proof}For regular POVM operators, we prove Eq.~(\ref{eq:uFku})
by generalizing the technique by Braunstein and Caves~\citep{braunstein_statistical_1994}.
The CFIM element corresponding
to a regular POVM $\Pi_{k}$ is defined as 
\begin{equation}
\mathcal{F}_{ij}^{k}(\bm{\lambda})=\partial_{i}\text{Tr}(\rho_{\bm{\lambda}}\Pi_{k})\partial_{j}\text{Tr}(\rho_{\bm{\lambda}}\Pi_{k})/\text{Tr}(\rho_{\bm{\lambda}}\Pi_{k}).
\end{equation}
Using Eq. (\ref{eq:SLD}) and the cyclic property of trace, i.e.,
\begin{equation}
\text{Tr}(L_{i\bm{\lambda}}\rho_{\boldsymbol{\lambda}}\Pi_{k})=\text{Tr}(\rho_{\bm{\lambda}}\Pi_{k}L_{i\bm{\lambda}})=[\text{Tr}(\rho_{\boldsymbol{\lambda}}L_{i\bm{\lambda}}\Pi_{k})]^{*},
\end{equation}
we obtain \citep{braunstein_statistical_1994,paris_quantum_2009}
\begin{align}
\partial_{i}\text{Tr}(\rho_{\boldsymbol{\lambda}}\Pi_{k}) & =\text{Tr}(\partial_{i}\rho_{\boldsymbol{\lambda}}\Pi_{k})\nonumber \\
 & =\frac{1}{2}[\text{Tr}(L_{i\bm{\lambda}}\rho_{\boldsymbol{\lambda}}\Pi_{k})+\text{Tr}(\rho_{\boldsymbol{\lambda}}L_{i\bm{\lambda}}\Pi_{k})]\nonumber \\
 & =\text{Re}[\text{Tr}(\rho_{\bm{\lambda}}\Pi_{k}L_{i\bm{\lambda}})].
\end{align}
Therefore for a real and nonzero vector $\bm{u}$, we obtain

\begin{align}
\sum_{ij}u_{i}\mathcal{F}_{ij}^{k}u_{j} & =\frac{[\text{Re}\text{Tr}(\rho_{\bm{\lambda}}\Pi_{k}\sum_{i}u_{i}L_{i\bm{\lambda}})]^{2}}{\text{Tr}(\rho_{\bm{\lambda}}\Pi_{k})}\nonumber \\
 & \leq\frac{\big|\text{Tr}(\rho_{\bm{\lambda}}\Pi_{k}\sum_{i}u_{i}L_{i\bm{\lambda}})\big|^{2}}{\text{Tr}(\rho_{\bm{\lambda}}\Pi_{k})}\nonumber \\
 & \leq\sum_{ij}u_{i}u_{j}\text{Tr}(\rho_{\bm{\lambda}}L_{i\bm{\lambda}}\Pi_{k}L_{j\bm{\lambda}})\nonumber \\
 & =\frac{1}{2}\sum_{ij}u_{i}u_{j}\left[\text{Tr}(\rho_{\bm{\lambda}}L_{i\bm{\lambda}}\Pi_{k}L_{j\bm{\lambda}})\right.\nonumber \\
 & \left.+\text{Tr}(\rho_{\bm{\lambda}}L_{j\bm{\lambda}}\Pi_{k}L_{i\bm{\lambda}})\right]\nonumber \\
 & =\sum_{ij}u_{i}u_{j}\text{Re}[\text{Tr}(\rho_{\bm{\lambda}}L_{i\bm{\lambda}}\Pi_{k}L_{j\bm{\lambda}})],\label{eq:bFrb2}
\end{align}
we have in the second inequality applied the Cauchy-Swartz inequality
$\big|\text{Tr}(A^{\dagger}B)\big|^{2}\leq\text{Tr}(A^{\dagger}A)\text{Tr}(B^{\dagger}B)$,
with $A\equiv\sqrt{\Pi_{k}}\sqrt{\rho_{\bm{\lambda}}}$ and $B\equiv\sum_{i}\sqrt{\Pi_{k}}u_{i}L_{i\bm{\lambda}}\sqrt{\rho_{\bm{\lambda}}}$.
Due to the fact $u_{i}u_{j}$ is symmetric in indices $i,\,j$, we
have symmetrized $\text{Tr}(\rho_{\bm{\lambda}}L_{i\bm{\lambda}}\Pi_{k}L_{j\bm{\lambda}})$
in the second last equality to obtain the last equality. \end{proof}

\subsection{\label{subsec:null}Proof of Eq. (\ref{eq:uFku}) for null POVM operators}

Before we start the proof, let us first prove an observation that
will be useful later.

\begin{lem}\label{lem:kerPi}A measurement operator $\Pi_{k}$ is
null if and only if $\forall n$, $\ket{\psi_{n\bm{\lambda}}}$ lies
in the kernel of $\Pi_{k}$, i.e.,
\begin{equation}
\Pi_{k}\ket{\psi_{n\bm{\lambda}}}=0\label{eq:PikPsi}
\end{equation}

\end{lem}

\begin{proof}The lemma is a consequence of the semi-positive definiteness
of $\rho_{\bm{\lambda}}$ and $\Pi_{k}$. To see this, let us note
a null measurement operator satisfies, 
\begin{equation}
\text{Tr}(\rho_{\bm{\lambda}}\Pi_{k})=\text{\ensuremath{\sum_{n}p_{n\bm{\lambda}}\braket{\psi_{n\bm{\lambda}}\big|\Pi_{k}\big|\psi_{n\bm{\lambda}}}=0}}\label{eq:TrrhoPi}
\end{equation}
Since $\rho_{\bm{\lambda}}$ is positive-definite on its support,
we conclude $p_{n\bm{\lambda}}$'s are strictly positive, as we defined
in Sec. \ref{sec:Preliminary}. Therefore Eq. (\ref{eq:TrrhoPi})
becomes
\begin{equation}
\braket{\psi_{n\bm{\lambda}}\big|\Pi_{k}\big|\psi_{n\bm{\lambda}}}=0\label{eq:PsiPikPsi}
\end{equation}
On the other hand, we may decompose $\ket{\psi_{n\bm{\lambda}}}$
into components that lie in the support and kernel of $\Pi_{k}$.
Then, since $\Pi_{k}$ is positive definite on its support, Eq. (\ref{eq:PsiPikPsi})
is equivalent as the fact that $\ket{\psi_{n\bm{\lambda}}}$ completely
lies in the kernel of $\Pi_{k}$.

\end{proof}

Now let us prove Eq. (\ref{eq:uFku}) for the case of null measurement
operators.

\begin{proof}Introducing short hand notation
\begin{equation}
g_{ij}^{k}(\bm{\lambda}^{\prime})\equiv\text{Tr}(\partial_{i}\rho_{\bm{\lambda}^{\prime}}\Pi_{k})\text{Tr}(\partial_{j}\rho_{\bm{\lambda}^{\prime}}\Pi_{k}),
\end{equation}
\begin{equation}
h^{k}(\bm{\lambda}^{\prime})\equiv\text{Tr}(\rho_{\bm{\lambda}^{\prime}}\Pi_{k}),
\end{equation}
the CFIM element $\mathcal{F}_{ij}^{k}$ corresponding to a null projector
$\Pi_{k}$ defined in the main text can be rewritten as 
\begin{equation}
\mathcal{F}_{ij}^{k}(\bm{\lambda})\equiv\lim_{\bm{\lambda}^{\prime}\to\bm{\lambda}}\frac{g_{ij}^{k}(\bm{\lambda}^{\prime})}{h^{k}(\bm{\lambda}^{\prime})}.\label{eq:Fijk-redef}
\end{equation}
Since the right hand side of Eq.~(\ref{eq:Fijk-redef}) is of the
type $0/0$, we need to Taylor expand both the numerator and denominator,
which will involve the derivatives of $\rho_{\bm{\lambda}}$. It is
straightforward to show that the first order derivatives of $g_{ij}^{k}(\bm{\lambda}^{\prime})$
and $h^{k}(\bm{\lambda}^{\prime})$ vanish at $\bm{\lambda}$, i.e.,
\begin{align}
\partial_{p}g_{ij}^{k}(\bm{\lambda}) & =\text{Tr}(\partial_{p}\partial_{i}\rho_{\bm{\lambda}}\Pi_{k})\cancel{\text{Tr}(\partial_{j}\rho_{\bm{\lambda}}\Pi_{k})}\\
 & +\cancel{\text{Tr}(\partial_{i}\rho_{\bm{\lambda}}\Pi_{k})}\text{Tr}(\partial_{p}\partial_{j}\rho_{\bm{\lambda}}\Pi_{k})=0,
\end{align}
\begin{equation}
\partial_{p}h^{k}(\bm{\lambda})=\text{Tr}(\partial_{p}\rho_{\bm{\lambda}}\Pi_{k})=0,
\end{equation}
due to Lemma \ref{lem:kerPi}. Therefore we need to expand $g_{ij}^{k}(\bm{\lambda}^{\prime})$
and $h^{k}(\bm{\lambda}^{\prime})$ to the second order in $\delta\bm{\lambda}\equiv\bm{\lambda}^{\prime}-\bm{\lambda}$
and calculate their second derivatives at $\bm{\lambda}$, i.e.,
\begin{align}
\partial_{p}\partial_{q}g_{ij}^{k}(\bm{\lambda}) & =\text{Tr}(\partial_{p}\partial_{i}\rho_{\bm{\lambda}}\Pi_{k})\text{Tr}(\partial_{q}\partial_{j}\rho_{\bm{\lambda}}\Pi_{k})\nonumber \\
 & +\text{Tr}(\partial_{q}\partial_{i}\rho_{\bm{\lambda}}\Pi_{k})\text{Tr}(\partial_{p}\partial_{j}\rho_{\bm{\lambda}}\Pi_{k}),\label{eq:dmdngijk}
\end{align}
\begin{equation}
\partial_{p}\partial_{q}h^{k}(\bm{\lambda})=\text{Tr}(\partial_{p}\partial_{q}\rho_{\bm{\lambda}}\Pi_{k}).\label{eq:dmdnhk}
\end{equation}
If we define
\begin{equation}
T_{ij}^{k}\equiv\frac{1}{2}\text{Tr}(\partial_{i}\partial_{j}\rho_{\bm{\lambda}}\Pi_{k}),\label{eq:Tijk-def}
\end{equation}
Substitution of Eq.~(\ref{eq:Tijk-def}) into Eqs.~(\ref{eq:dmdngijk},
\ref{eq:dmdnhk}) gives
\begin{equation}
\partial_{p}\partial_{q}g_{ij}^{k}(\bm{\lambda})=4(T_{pi}^{k}T_{qj}^{k}+T_{qi}^{k}T_{pj}^{k}),\label{eq:dpdqgijk}
\end{equation}
\begin{equation}
\partial_{p}\partial_{q}h^{k}(\bm{\lambda}^{\prime})=2T_{pq}^{k}.\label{eq:dpdqhk}
\end{equation}
Substituting $g_{ij}^{k}(\bm{\lambda}^{\prime})=\sum_{p,\,q}\partial_{p}\partial_{q}g_{ij}^{k}(\bm{\lambda})\delta\lambda_{p}\delta\lambda_{q}$
and $h^{k}(\bm{\lambda}^{\prime})=\sum_{p,\,q}\partial_{p}\partial_{q}h^{k}(\bm{\lambda})\delta\lambda_{p}\delta\lambda_{q}$
into Eq.~(\ref{eq:Fijk-redef}), with notice of Eqs.~(\ref{eq:dpdqgijk},
\ref{eq:dpdqhk}), we arrive at
\begin{equation}
\mathcal{F}_{ij}^{k}(\bm{\lambda})=\frac{2\sum_{pq}(T_{pi}^{k}T_{qj}^{k}+T_{qi}^{k}T_{pj}^{k})\delta\lambda_{p}\delta\lambda_{q}}{\sum_{pq}T_{pq}^{k}\delta\lambda_{p}\delta\lambda_{q}}.\label{eq:Fijk-null}
\end{equation}
According to Eq.~(\ref{eq:SLD}) we find\begin{widetext}

\begin{align}
\text{Re}\text{Tr}(\rho_{\bm{\lambda}}L_{i\bm{\lambda}}\Pi_{k}L_{j\bm{\lambda}}) & =\frac{1}{2}[\text{Tr}(\rho_{\bm{\lambda}}L_{i\bm{\lambda}}\Pi_{k}L_{j\bm{\lambda}})+\text{c.c.}]=\frac{1}{2}[\text{Tr}(L_{j\bm{\lambda}}\rho_{\bm{\lambda}}L_{i\bm{\lambda}}\Pi_{k})+\text{Tr}(L_{i\bm{\lambda}}\rho_{\bm{\lambda}}L_{j\bm{\lambda}}\Pi_{k})]\nonumber \\
 & =\text{Tr}(L_{j\bm{\lambda}}\partial_{i}\rho_{\bm{\lambda}}\Pi_{k})+\text{Tr}(\partial_{i}\rho_{\bm{\lambda}}L_{j\bm{\lambda}}\Pi_{k})-\frac{1}{2}\cancel{\text{Tr}(L_{j\bm{\lambda}}L_{i\bm{\lambda}}\rho_{\bm{\lambda}}\Pi_{k})}-\frac{1}{2}\cancel{\text{Tr}(L_{i\bm{\lambda}}L_{j\bm{\lambda}}\rho_{\bm{\lambda}}\Pi_{k})}\nonumber \\
 & =\text{Tr}(\partial_{i}(L_{j\bm{\lambda}}\rho_{\bm{\lambda}})\Pi_{k})+\text{Tr}(\partial_{i}(\rho_{\bm{\lambda}}L_{j\bm{\lambda}})\Pi_{k})-\cancel{\text{Tr}(\partial_{i}L_{j\bm{\lambda}}\rho_{\bm{\lambda}}\Pi_{k})}-\cancel{\text{Tr}(\rho_{\bm{\lambda}}\partial_{i}L_{j\bm{\lambda}}\Pi_{k})}\nonumber \\
 & =\text{Tr}(\partial_{i}(L_{j\bm{\lambda}}\rho_{\bm{\lambda}}+\rho_{\bm{\lambda}}L_{j\bm{\lambda}})\Pi_{k})=2\text{Tr}(\partial_{i}\partial_{j}\rho_{\bm{\lambda}}\Pi_{k}),
\end{align}
where the cancellations of the terms are due to Lemma \ref{lem:kerPi}.
In view of Eq.~(\ref{eq:Tijk-def}), we arrive at 
\begin{equation}
T_{ij}^{k}=\frac{1}{4}\text{Re}\text{Tr}(\rho_{\bm{\lambda}}L_{i\bm{\lambda}}\Pi_{k}L_{j\bm{\lambda}})=\frac{1}{4}\mathcal{I}_{ij}^{k}.\label{eq:Tijk-reexpr}
\end{equation}
We first derive the following inequality for later use. For real and
non-zero $\delta\bm{\lambda}$ and $\bm{u}$, we obtain

\begin{align}
(\sum_{ij}\delta\lambda_{i}T_{ij}^{k}u_{j})^{2} & =\frac{1}{16}(\text{Re}[\sum_{ij}\delta\lambda_{i}\text{Tr}(\rho_{\bm{\lambda}}L_{i\bm{\lambda}}\Pi_{k}L_{j\bm{\lambda}})u_{j}])^{2}\nonumber \\
 & \le\frac{1}{16}\big|\sum_{ij}\delta\lambda_{i}\text{Tr}(\rho_{\bm{\lambda}}L_{i\bm{\lambda}}\Pi_{k}L_{j\bm{\lambda}})u_{j}\big|^{2}\nonumber \\
 & =\frac{1}{16}\big|\text{Tr}(\sum_{i}\delta\lambda_{i}\sqrt{\rho_{\bm{\lambda}}}L_{i\bm{\lambda}}\sqrt{\Pi_{k}}\sum_{j}u_{j}\sqrt{\Pi_{k}}L_{j\bm{\lambda}}\sqrt{\rho_{\bm{\lambda}}})\big|^{2}\nonumber \\
 & \le\frac{1}{16}\left[\sum_{ij}\delta\lambda_{i}\delta\lambda_{j}\text{Tr}(\sqrt{\rho_{\bm{\lambda}}}L_{i\bm{\lambda}}\sqrt{\Pi_{k}}\sqrt{\Pi_{k}}L_{j\bm{\lambda}}\sqrt{\rho_{\bm{\lambda}}})\right]\left[\sum_{ij}u_{i}u_{j}\text{Tr}(\sqrt{\rho_{\bm{\lambda}}}L_{i\lambda}\sqrt{\Pi_{k}}\sqrt{\Pi_{k}}L_{j\bm{\lambda}}\sqrt{\rho_{\bm{\lambda}}})\right]\nonumber \\
 & =\left(\sum_{ij}T_{ij}^{k}\delta\lambda_{i}\delta\lambda_{j}\right)\left(\sum_{ij}T_{ij}^{k}u_{i}u_{j}\right)\label{eq:del-lam-Tij-uj}
\end{align}
\end{widetext}where we have used the Cauchy-Swartz inequality $\big|\text{Tr}(A^{\dagger}B)\big|^{2}\leq\text{Tr}(A^{\dagger}A)\text{Tr}(B^{\dagger}B)$,
with $A\equiv\sum_{i}\delta\lambda_{i}\sqrt{\Pi_{k}}L_{i\bm{\lambda}}\sqrt{\rho_{\bm{\lambda}}}$
and $B\equiv\sum_{i}\sqrt{\Pi_{k}}u_{i}L_{i\bm{\lambda}}\sqrt{\rho_{\bm{\lambda}}}$
in the second inequality and performed symmetrization to obtain the
last equality. Note that the denominator of Eq.~(\ref{eq:Fijk-null})
can be rewritten as, upon substitution of Eq.~(\ref{eq:Tijk-reexpr}),
\begin{align}
\sum_{pq}T_{pq}^{k}\delta\lambda_{p}\delta\lambda_{q} & =\text{Re}\text{Tr}[\rho_{\bm{\lambda}}\delta\Lambda\Pi_{k}\delta\Lambda]\nonumber \\
 & =\text{Tr}[\sqrt{\Pi_{k}}\delta\Lambda\rho_{\bm{\lambda}}(\sqrt{\Pi_{k}}\delta\Lambda)^{\dagger}]\ge0,
\end{align}
where $\delta\Lambda\equiv\sum_{p}L_{p\bm{\lambda}}\delta\lambda_{p}$.
Therefore for any $\delta\bm{\lambda}$ the denominator of Eq.~(\ref{eq:Fijk-null})
is non-negative. With these observations, next we find 
\begin{alignat}{1}
\sum_{ij}u_{i}\mathcal{F}_{ij}^{k}u_{j} & =\frac{4\sum_{ij}u_{i}u_{j}\sum_{pq}T_{pi}^{k}T_{qj}^{k}\delta\lambda_{p}\delta\lambda_{q}}{\sum_{pq}T_{pq}^{k}\delta\lambda_{p}\delta\lambda_{q}}\nonumber \\
 & =\frac{4(\sum_{pi}\delta\lambda_{p}T_{pi}^{k}u_{i})^{2}}{\sum_{pq}T_{pq}^{k}\delta\lambda_{p}\delta\lambda_{q}}\nonumber \\
 & \leq\frac{4(\sum_{pq}T_{pq}^{k}\delta\lambda_{p}\delta\lambda_{q})\sum_{ij}u_{i}u_{j}T_{ij}^{k}}{\sum_{pq}T_{pq}^{k}\delta\lambda_{p}\delta\lambda_{q}}\nonumber \\
 & =\sum_{ij}u_{i}u_{j}\mathcal{I}_{ij}^{k},\label{eq:bF0b2}
\end{alignat}
where we have used Eq.~(\ref{eq:del-lam-Tij-uj}) in the inequality
to get the upper bound.\end{proof}

\subsection{Equivalence between the saturations of the scalar and matrix Helstrom
CR bounds}

Remarkably, based on the following Lemma, the saturation of the matrix
bound (\ref{eq:HelstromCR}) is equivalent as the saturation of the
corresponding scalar bound with respect to any given positive definite
weight matrix, which has been used in previous papers \citep{vidrighin2014jointestimation,ragy2016compatibility}.

\begin{lem}\label{lem:matrix-scalar}Given two positive semi-definite
matrices $A$, $B$ and a positive definite weight matrix $G$ called
weight matrix, if $A\curlyeqprec B$, i.e., $B-A$ is positive semi-definite,
then the following statements are equivalent

(i) $\exists$ a positive definite weight matrix $G_{0}$ such that
$\text{Tr}(G_{0}A)=\text{Tr}(G_{0}B)$

(ii) $A=B$ 

(iii) $\forall$positive definite weight matrix $G$, $\text{Tr}(GA)=\text{Tr}(GB)$
holds

(iv) $\exists$ a positive definite weight matrix $G_{0}$ such that
$\text{Tr}(G_{0}A^{-1})=\text{Tr}(G_{0}B^{-1})$

(v) $\forall$positive definite matrix $G$, $\text{Tr}(GA^{-1})=\text{Tr}(GB^{-1})$
holds\end{lem}

\begin{proof}We notice that once the $\text{(i)}\Longleftrightarrow\text{(ii)}$
is justified, proof of the equivalence between any statements becomes
straightforward. To prove $\text{(i)}\Longleftrightarrow\text{(ii)}$,
it is sufficient to show that $\text{(i)}\Longrightarrow\text{(ii)}$
as the opposite direction is trivial. Condition (i) implies that $\text{Tr}[G_{0}(B-A)]=0$.
Now moving to the basis that diagonalizes the matrix $B-A$, i.e.,
$B-A=U^{\dagger}DU$, where $D$ is diagonal, we find 
\begin{equation}
\text{Tr}[G_{0}(B-A)]=\text{Tr}(U^{\dagger}G_{0}UD)=\text{\ensuremath{\sum}}_{n}\tilde{G}_{0nn}D_{nn}=0\label{eq:TrGB-A}
\end{equation}
where $\tilde{G}_{0}\equiv U^{\dagger}G_{0}U$ is representation of
$G_{0}$ in the basis that diagonalizes $B-A$. Since $A\preceq B$,
$B-A$ is also semi-positive definite and therefore $D_{nn}\ge0$.
On the other hand $G_{0}$ is positive definite, which indicates diagonal
matrix element of its representation in every basis is positive. Then
we know $\tilde{G}_{0nn}>0$. The only way to satisfy Eq. (\ref{eq:TrGB-A})
is $D_{nn}=0$ for all $n$. Therefore we conclude that $A=B$.\end{proof}

In Lemma \ref{lem:matrix-scalar}, if we take $A$ as the CFIM and
$B$ as the QFIM, we immediately conclude that the saturation of the
matrix Helstrom bound is equivalent as the scalar bound with respect
to any given positive definite weight matrix.

\section{\textup{Saturation of the Helstrom CR bound and the partial commutativity
condition}}

The physical implications of $\mathcal{F}_{ij}^{k}$ and $\mathcal{I}_{ij}^{k}$
are very important in understanding the saturation of the Helstrom
CR bound: from Eq.~(\ref{eq:uFku}), we see that for each POVM operator
$\Pi_{k}$, either regular or null, the corresponding QFIM $\mathcal{I}^{k}$
is a matrix bound for the corresponding CFIM $\mathcal{F}^{k}$. The
saturation of the Helstrom CR bound requires the saturation of all
these matrix bounds. Following this idea, we can\textit{ }derive the
saturation conditions for the Helstrom CR bound by saturating Eq.~(\ref{eq:uFku})
for regular and null POVM operators respectively.

\subsection{Saturation conditions for general POVMs}

\begin{thm} \label{thm:reg}The matrix bound of the CFIM due to a
\textit{regular} operator $\Pi_{k}$ is saturated at $\bm{\lambda}$,
if and only if
\begin{equation}
\Pi_{k}L_{i\bm{\lambda}}\ket{\psi_{n\bm{\lambda}}}=\xi_{i}^{k}\Pi_{k}\ket{\psi_{n\bm{\lambda}}},\,\forall i,\,k,\,n\label{eq:sat-cond-reg}
\end{equation}
where $\xi_{i}^{k}$ is real and independent of $n$.\end{thm}

\begin{proof}The saturation of the first inequality of Eq.~(\ref{eq:bFrb2})
requires that $\text{Tr}(\rho_{\bm{\lambda}}\Pi_{k}\sum_{i}u_{i}L_{i\bm{\lambda}})$
must be real for any arbitrary real and nonzero vector $\bm{u}$.
Therefore $\text{Tr}(\rho_{\bm{\lambda}}\Pi_{k}L_{i\bm{\lambda}})$
must be real for each $i$. The saturation of the second inequality
of Eq.~(\ref{eq:bFrb2}) requires that $\sqrt{\Pi_{k}}\sqrt{\rho_{\bm{\lambda}}}$
must be proportional to $\sqrt{\Pi_{k}}\sum_{i}u_{i}L_{i\bm{\lambda}}\sqrt{\rho_{\bm{\lambda}}}$
for any arbitrary, non-zero and real $\bm{u}$. Thus $\sqrt{\Pi_{k}}\sqrt{\rho_{\bm{\lambda}}}$
must be proportional to $\sqrt{\Pi_{k}}L_{i\bm{\lambda}}\sqrt{\rho_{\bm{\lambda}}}$
for each $i$, i.e., 
\begin{equation}
\xi_{i}^{k}\sqrt{\Pi_{k}}\sqrt{\rho_{\bm{\lambda}}}=\sqrt{\Pi_{k}}L_{i\bm{\lambda}}\sqrt{\rho_{\bm{\lambda}}},\,\forall i.\label{eq:CS-eq-reg}
\end{equation}
Eq. (\ref{eq:CS-eq-reg}) can be rewritten as
\begin{equation}
\sqrt{\Pi_{k}}(L_{i\bm{\lambda}}-\xi_{i}^{k})\sqrt{\rho_{\bm{\lambda}}}=0,\,\forall i.\label{eq:CS-eq-reg2}
\end{equation}
Since $\rho_{\bm{\lambda}}$ and $\sqrt{\rho_{\bm{\lambda}}}$ has
the same kernel, Eq. (\ref{eq:CS-eq-reg2}) is equivalent as $\sqrt{\Pi_{k}}(L_{i\bm{\lambda}}-\xi_{i}^{k})\rho_{\bm{\lambda}}=0$,
which is equivalent to 
\begin{equation}
\Pi_{k}L_{i\bm{\lambda}}\rho_{\bm{\lambda}}=\xi_{i}^{k}\Pi_{k}\rho_{\bm{\lambda}}
\end{equation}
due to the same reason. With the spectral decomposition $\rho_{\bm{\lambda}}=\sum_{n}p_{n\bm{\lambda}}\ket{\psi_{n\bm{\lambda}}}\bra{\psi_{n\bm{\lambda}}}$,
we arrive at Eq. (\ref{eq:sat-cond-reg}). The proportionality constant
$\xi_{i}^{k}$ can be found by taking the trace inner product with
$\sqrt{\Pi_{k}}\sqrt{\rho_{\bm{\lambda}}}$ on both sides of Eq.~(\ref{eq:CS-eq-reg}),
i.e.,
\begin{equation}
\xi_{i}^{k}=\text{Tr}(\rho_{\bm{\lambda}}\Pi_{k}L_{i\bm{\lambda}})/\text{Tr}(\rho_{\bm{\lambda}}\Pi_{k}),\,\forall i.\label{eq:xi-xk}
\end{equation}
Therefore, the condition of $\text{Tr}(\rho_{\bm{\lambda}}\Pi_{k}L_{i\bm{\lambda}})$
being real is equivalent to that $\xi_{i}^{k}$ is real, which concludes
the proof. \end{proof}

\begin{thm} \label{thm:null}The matrix bound of the CFIM due to
a \textit{null} operator $\Pi_{k}$ is saturated at $\bm{\lambda}$,
if and only if

\begin{equation}
\Pi_{k}L_{i\bm{\lambda}}\ket{\psi_{n\bm{\lambda}}}=\eta_{ij}^{k}\Pi_{k}L_{j\bm{\lambda}}\ket{\psi_{n\bm{\lambda}}},\,\forall i,\,j,\,k,\,n\label{eq:sat-cond-null}
\end{equation}
where the $\eta_{ij}^{k}$ is real and independent of $n$. \end{thm}

\begin{proof}Since the saturation of Eq.~(\ref{eq:bF0b2}) is equivalent
as the saturation of the two inequalities in Eq.~(\ref{eq:del-lam-Tij-uj}),
we will work with Eq.~(\ref{eq:del-lam-Tij-uj}) subsequently. The
saturation of the first inequality in Eq.~(\ref{eq:del-lam-Tij-uj})
requires $\sum_{ij}\delta\lambda_{i}\text{Tr}(\rho_{\bm{\lambda}}L_{i\bm{\lambda}}\Pi_{k}L_{j\bm{\lambda}})u_{j}$
is real for any $\delta\bm{\lambda},\,\bm{u}$. This indicates that
$\text{Tr}(\rho_{\bm{\lambda}}L_{i\bm{\lambda}}\Pi_{k}L_{j\bm{\lambda}})$
must be real for any pair $i,\,j$. The saturation of the second inequality
in Eq.~(\ref{eq:del-lam-Tij-uj}) requires $\sum_{i}\delta\lambda_{i}\sqrt{\Pi_{k}}L_{i\bm{\lambda}}\sqrt{\rho_{\bm{\lambda}}}$
be proportional to $\sum_{j}u_{j}\sqrt{\Pi_{k}}L_{j\bm{\lambda}}\sqrt{\rho_{\bm{\lambda}}}$
for any $\delta\bm{\lambda},\,\bm{u}$. Thus $\sqrt{\Pi_{k}}L_{i\bm{\lambda}}\sqrt{\rho_{\bm{\lambda}}}$
must be proportional to $\sqrt{\Pi_{k}}L_{j\bm{\lambda}}\sqrt{\rho_{\bm{\lambda}}}$
for each pair $(i,\,j)$, i.e.,
\begin{equation}
\sqrt{\Pi_{k}}L_{i\bm{\lambda}}\sqrt{\rho_{\bm{\lambda}}}=\eta_{ij}^{k}\sqrt{\Pi_{k}}L_{j\bm{\lambda}}\sqrt{\rho_{\bm{\lambda}}},\,\forall i,\,j,\,k.\label{eq:CS-eq-null}
\end{equation}
Since $\rho_{\bm{\lambda}}$ and $\sqrt{\rho_{\bm{\lambda}}}$ has
the same kernel, Eq. (\ref{eq:CS-eq-null}) is equivalent as $\sqrt{\Pi_{k}}(L_{i\bm{\lambda}}-\eta_{ij}^{k}L_{j\bm{\lambda}})\rho_{\bm{\lambda}}=0$,
which is equivalent to 
\begin{equation}
\Pi_{k}L_{i\bm{\lambda}}\rho_{\bm{\lambda}}=\eta_{ij}^{k}\Pi_{k}L_{j\bm{\lambda}}\rho_{\bm{\lambda}}
\end{equation}
due to the same reason. With the spectral decomposition $\rho_{\bm{\lambda}}=\sum_{n}p_{n\bm{\lambda}}\ket{\psi_{n\bm{\lambda}}}\bra{\psi_{n\bm{\lambda}}}$,
we arrive at Eq. (\ref{eq:sat-cond-null}). The proportionality constant
$\eta_{ij}^{k}$ can be found by taking the trace inner product with
$\sqrt{\Pi_{k}}L_{j\bm{\lambda}}\sqrt{\rho_{\bm{\lambda}}}$ on both
sides of Eq.~(\ref{eq:CS-eq-reg}), i.e.,
\begin{equation}
\eta_{ij}^{k}=\frac{\text{Tr}(\rho_{\bm{\lambda}}L_{j\bm{\lambda}}\Pi_{k}L_{i\bm{\lambda}})}{\text{Tr}(\rho_{\bm{\lambda}}L_{j\bm{\lambda}}\Pi_{k}L_{j\bm{\lambda}})}
\end{equation}
Since $\text{Tr}(\rho_{\bm{\lambda}}L_{j\bm{\lambda}}\Pi_{k}L_{j\bm{\lambda}})$
is real, the condition that $\text{Tr}(\rho_{\bm{\lambda}}L_{i\bm{\lambda}}\Pi_{k}L_{j\bm{\lambda}})$
must be real is equivalent to that $\eta_{ij}^{k}$ must be real,
which concludes the proof. \end{proof}

\subsection{The partial commutativity condition}

Up to now Theorems \ref{thm:reg}-\ref{thm:null} focus on the exact
saturation of the QFIM. For limited experimental situations, where
only separable measurements or collective measurements that involve
only a small number of states are experimentally implementable, the
CFIM scales exactly linearly with the number of identical states.
As a result, the QFIM can only be saturated exactly. We show in the
following theorem that saturating the QFIM \textit{exactly} the following
theorem must hold.

\begin{thm}(\textit{The partial commutativity condition}) To \textit{exactly}
saturate the QFIM, it is necessary to have the SLD must commute on
the support of $\rho_{\bm{\lambda}}$, i.e., $\braket{\psi_{m\bm{\lambda}}\big|[L_{i\bm{\lambda}},\,L_{j\bm{\lambda}}]\big|\psi_{n\bm{\lambda}}}=0,\,\forall n,\,m$.

\end{thm}

\begin{proof}If $\Pi_{k}$ is regular, according to Eq.~(\ref{eq:sat-cond-reg}),
we obtain
\begin{align}
\braket{\psi_{m\bm{\lambda}}\big|L_{i\bm{\lambda}}\Pi_{k}L_{j\bm{\lambda}}\big|\psi_{n\bm{\lambda}}} & =\xi_{j}^{k}\braket{\psi_{m\bm{\lambda}}\big|L_{i\bm{\lambda}}\Pi_{k}\big|\psi_{n\bm{\lambda}}}\nonumber \\
 & =\xi_{i}^{k}\xi_{j}^{k}\braket{\psi_{m\bm{\lambda}}\big|\Pi_{k}\big|\psi_{n\bm{\lambda}}}\nonumber \\
 & =\braket{\psi_{m\bm{\lambda}}\big|L_{j\bm{\lambda}}\Pi_{k}L_{i\bm{\lambda}}\big|\psi_{n\bm{\lambda}}}
\end{align}
If $\Pi_{k}$ is null, according to Eq.~(\ref{eq:sat-cond-null}),
we obtain 
\begin{align}
\braket{\psi_{m\bm{\lambda}}\big|L_{i\bm{\lambda}}\Pi_{k}L_{j\bm{\lambda}}\big|\psi_{n\bm{\lambda}}} & =\eta_{ij}^{k}\braket{\psi_{m\bm{\lambda}}\big|L_{j\bm{\lambda}}\Pi_{k}L_{j\bm{\lambda}}\big|\psi_{n\bm{\lambda}}}\nonumber \\
 & =\braket{\psi_{m\bm{\lambda}}\big|L_{j\bm{\lambda}}\eta^{ijk}\Pi_{k}L_{j\bm{\lambda}}\big|\psi_{n\bm{\lambda}}}\nonumber \\
 & =\braket{\psi_{m\bm{\lambda}}\big|L_{j\bm{\lambda}}\Pi_{k}L_{i\bm{\lambda}}\big|\psi_{n\bm{\lambda}}}
\end{align}
Thus we can see that the optimal measurement mediates the commutativity
between $L_{i\bm{\lambda}}$ and $L_{j\bm{\lambda}}$ on the support
of $\rho_{\bm{\lambda}}$, i.e., for any $m$ and $n$,
\begin{align}
\braket{\psi_{m\bm{\lambda}}\big|[L_{i\bm{\lambda}},\,L_{j\bm{\lambda}}]\big|\psi_{n\bm{\lambda}}} & =\sum_{k}\braket{\psi_{m\bm{\lambda}}\big|L_{i\bm{\lambda}}\Pi_{k}L_{j\bm{\lambda}}\big|\psi_{n\bm{\lambda}}}\nonumber \\
 & -\sum_{k}\braket{\psi_{m\bm{\lambda}}\big|L_{j\bm{\lambda}}\Pi_{k}L_{i\bm{\lambda}}\big|\psi_{n\bm{\lambda}}}=0
\end{align}
 \end{proof}

The partial commutativity condition reduces to the weak commutativity
condition \citep{matsumoto2002anew} for pure states and the full
commutativity condition for full rank states. While the sufficiency
of the partial commutativity to saturate the bound has been proved
to be true in the case of pure states \citep{matsumoto2002anew,pezz`e2017optimal}
and is trivially true in the case of the full rank states, whether
it is still true in the general case is beyond the scope of the current
work.

When collective measurements that entangle a large number of states
are allowed, where the CFIM has the sublinear corrections in general
besides the linear term with respect to the number of identical copies,
asymptotic saturation becomes relevant. However, we emphasize that
our Theorems \ref{thm:reg}-\ref{thm:null} should be asymptotically
satisfied if QFIM is asymptotically saturated. Ref. \citep{ragy2016compatibility}
concludes that if all possible measurements are allowed to perform
on the large number of states, the weak commutativity condition $\text{Tr}(\rho_{\bm{\lambda}}[L_{i\bm{\lambda}},\,L_{j\bm{\lambda}}])=0$
is necessary and sufficient for the saturation of the Helstrom bound.
However, the precise connection between our formalism here and Ref.
\citep{ragy2016compatibility} is still an open question.

\subsection{Alternative saturation conditions}

When one deals with specific problems, e.g., the problem of superresolution
of two incoherent optical point sources in Sec. \ref{sec:superresolution},
it is useful to consider the spectral decomposition of the POVM operator.
Combining with the matrix representation of $L_{i\bm{\lambda}}$ given
in Appendix \ref{sec:Matrix-SLD}, Theorems \ref{thm:reg}, \ref{thm:null}
can be alternatively rewritten as follows:

\begin{thm}\label{thm:reg-explicit}The matrix bound of the CFIM
due to a \textit{regular} operator $\Pi_{k}=\sum_{\alpha}q_{k\alpha}\ket{\pi_{k\alpha}}\bra{\pi_{k\alpha}}$
is saturated at $\bm{\lambda}$, if and only if
\begin{equation}
\braket{\psi_{n\bm{\lambda}}\big|L_{i\bm{\lambda}}^{\perp}\big|\pi_{k\alpha}}+2\braket{\partial_{i}^{0}\psi_{n\bm{\lambda}}\big|\pi_{k\alpha}}=\xi_{i}^{k}\braket{\psi_{n\bm{\lambda}}\big|\pi_{k\alpha}}\;\forall i,\,n,\,\alpha,\label{eq:reg-explicit}
\end{equation}
where $\xi_{i}^{k}$ is real and independent of $n$ and $\alpha$,
and $L_{i\bm{\lambda}}^{\perp}$ defined as Eq. (\ref{eq:Liperp})
denotes the projection of $L_{i\bm{\lambda}}$ onto $\text{supp}(\rho_{\bm{\lambda}})$.\end{thm}

\begin{proof}For a regular POVM operator $\Pi_{k}$ with a spectral
decomposition $\Pi_{k}=\sum_{\alpha}q_{k\alpha}\ket{\pi_{k\alpha}}\bra{\pi_{k\alpha}}$
where $q_{k\alpha}$ is strictly positive, Eq. (\ref{eq:sat-cond-reg})
becomes 
\begin{gather}
\sum_{\alpha}q_{k\alpha}\ket{\pi_{k\alpha}}\bra{\pi_{k\alpha}}L_{i\bm{\lambda}}\ket{\psi_{n\bm{\lambda}}}=\nonumber \\
\sum_{\alpha}q_{k\alpha}\xi_{i}^{k}\ket{\pi_{k\alpha}}\braket{\pi_{k\alpha}\big|\psi_{n\bm{\lambda}}},\:\forall i,\,n.
\end{gather}
Since $\ket{\pi_{k\alpha}}$'s are linearly independent, the above
equation is equivalent to 
\begin{equation}
\braket{\pi_{k\alpha}\big|L_{i\bm{\lambda}}\big|\psi_{n\bm{\lambda}}}=\xi_{i}^{k}\braket{\pi_{k\alpha}\big|\psi_{n\bm{\lambda}}},\:\forall i,\,\alpha,\,n.
\end{equation}
with $\xi_{i}^{k}$ being real and independent of $n$. According
to Eq.~(\ref{eq:LiPsiPerp}), we have 
\begin{equation}
\braket{\psi_{n\bm{\lambda}}\big|L_{i\bm{\lambda}}\big|\pi_{k\alpha}}=\braket{\psi_{n\bm{\lambda}}\big|L_{i\bm{\lambda}}^{\perp}\big|\pi_{k\alpha}}+2\braket{\partial_{i}^{0}\psi_{n\bm{\lambda}}\big|\pi_{k\alpha}}
\end{equation}
which concludes the proof. \end{proof}

\begin{thm}\label{thm:null-explicit}The matrix bound of the CFIM
due to a \textit{null} operator $\Pi_{k}=\sum_{\alpha}q_{k\alpha}\ket{\pi_{k\alpha}}\bra{\pi_{k\alpha}}$
is saturated at $\bm{\lambda}$, if and only if
\begin{equation}
\braket{\partial_{i}\tilde{\psi}_{n\bm{\lambda}}\big|\pi_{k\alpha}}=\eta_{ij}^{k}\braket{\partial_{j}\tilde{\psi}_{n\bm{\lambda}}\big|\pi_{k\alpha}}\forall i,\,j,\,n,\alpha,\label{eq:null-explicit}
\end{equation}
where $\ket{\tilde{\psi}_{n\bm{\lambda}}}$ is not necessarily normalized
and $\eta_{ij}^{k}$ is real and independent of $n$ and $\alpha$.\end{thm}

\begin{proof}For a null POVM operator $\Pi_{k}=\sum_{\alpha}q_{k\alpha}\ket{\pi_{k\alpha}}\bra{\pi_{k\alpha}}$
, according to Eq.~(\ref{eq:LiPsiPerp}), Eq. (\ref{eq:sat-cond-null})
becomes
\begin{gather}
\sum_{\alpha}q_{k\alpha}\ket{\pi_{k\alpha}}\braket{\pi_{k\alpha}\big|L_{i\bm{\lambda}}\big|\psi_{n\bm{\lambda}}}\nonumber \\
=\sum_{\alpha}\eta_{ij}^{k}q_{k\alpha}\ket{\pi_{k\alpha}}\braket{\pi_{k\alpha}\big|L_{j\bm{\lambda}}\big|\psi_{n\bm{\lambda}}}\label{eq:sat-null-explicit-0}
\end{gather}
Since $\ket{\pi_{k\alpha}}$'s are linearly independent, the above
equation is equivalent to 
\begin{equation}
\braket{\pi_{k\alpha}\big|L_{i\bm{\lambda}}\big|\psi_{n\bm{\lambda}}}=\eta_{ij}^{k}\braket{\pi_{k\alpha}\big|L_{j\bm{\lambda}}\big|\psi_{n\bm{\lambda}}},\:\forall i,\,\alpha,\,n.
\end{equation}
with $\eta_{ij}^{k}$ being real and independent of $n$. Since $\Pi_{k}$
is null, $\braket{\pi_{k\alpha}\big|\psi_{n\bm{\lambda}}}=0,\,\forall n,\,\alpha$.
Therefore, according to Eq.~(\ref{eq:LiPsiPerp}), we have 
\begin{equation}
\braket{\psi_{n\bm{\lambda}}\big|L_{i\bm{\lambda}}\big|\pi_{k\alpha}}=2\braket{\partial_{i}^{0}\psi_{n\bm{\lambda}}\big|\pi_{k\alpha}}.
\end{equation}
Therefore according to Eq. (\ref{eq:sat-null-explicit-0}), we arrive
at 
\begin{equation}
\braket{\partial_{i}^{0}\psi_{n\bm{\lambda}}\big|\pi_{k\alpha}}=\eta_{ij}^{k}\braket{\partial_{j}^{0}\psi_{n\bm{\lambda}}\big|\pi_{k\alpha}}\label{eq:sat-null-normalized}
\end{equation}
For a null operator $\Pi_{k}=\sum_{\alpha}q_{k\alpha}\ket{\pi_{k\alpha}}\bra{\pi_{k\alpha}}$,
due to Lemma \ref{lem:kerPi}, we have $\braket{\psi_{n\bm{\lambda}}\big|\pi_{k\alpha}}=0,\,\forall n,\alpha$.
This implies
\begin{equation}
\braket{\partial_{i}^{0}\psi_{n\bm{\lambda}}\big|\pi_{k\alpha}}=\braket{\partial_{i}\psi_{n\bm{\lambda}}\big|\pi_{k\alpha}}
\end{equation}
\begin{equation}
\braket{\partial_{i}\tilde{\psi}_{n\bm{\lambda}}\big|\pi_{k\alpha}}=\braket{\partial_{i}\psi_{n\bm{\lambda}}\big|\pi_{k\alpha}}\braket{\tilde{\psi}_{n\bm{\lambda}}\big|\tilde{\psi}_{n\bm{\lambda}}},
\end{equation}
where $\ket{\tilde{\psi}_{n\bm{\lambda}}}$ is an unnormalized state.
Upon substituting the above two equations into Eq.~(\ref{eq:sat-null-normalized}),
one can conclude the proof easily. \end{proof}

We emphasize that at the critical point of the change of the rank
of $\rho_{\bm{\lambda}}$, where there exists $\ket{\tilde{\psi}_{n\bm{\lambda}}}$
that locally vanishes at $\bm{\lambda}$ and therefore $\ket{\psi_{n\bm{\lambda}}}$
is not well-defined. While Eq.~(\ref{eq:null-explicit}) is still
valid in this case, Eq.~(\ref{eq:reg-explicit}) should be understood
in the sense of taking the limit $\bm{\lambda}^{\prime}\to\bm{\lambda}$
on both sides. We will illustrate this issue in the example of superresolution
subsequently, but will not delve into rigorous mathematical discussion
on the removable discontinuity of the QFIM at the critical point \citep{vsafranek2017discontinuities,rezakhani2015continuity}.
Note that $\eta_{ii}^{k}=1$, thus for a null projector in single
parameter estimation Eq.~(\ref{eq:null-explicit}) is satisfied automatically.
Moreover, $\eta_{ij}^{k}=1/\eta_{ji}^{k}$. Thus in order to check
whether a null projector satisfies Theorem \ref{thm:null-explicit},
one only needs to verify whether the upper or lower (excluding the
diagonal) matrix elements of $\eta^{k}$ are real and independent
of $n$ and $\alpha$.

\subsection{\label{subsec:pezz}Recovering the results by Pezz\`e et al. \citep{pezz`e2017optimal}}

Pezz\`e et al. \citep{pezz`e2017optimal} obtained the necessary
and sufficient conditions for a projective measurement consisting
of rank one projectors to saturate the Helstrom CR bound for pure
states. Here we recover their results through Theorems \ref{thm:reg-explicit}
and \ref{thm:null-explicit}. For a pure state $\rho_{\bm{\lambda}}=\ket{\psi_{\bm{\lambda}}}\bra{\psi_{\bm{\lambda}}}$,
the dimension of $\text{supp}(\rho_{\bm{\lambda}})$ is one. Therefore
$\xi_{i}^{k}$'s naturally do not depend on the index $n$. Furthermore,
for a rank one projector, we suppress the subscript $\alpha$ in the
basis vector $\ket{\pi_{k\alpha}}$ and observe that $\xi_{i}^{k}$'s
naturally do not $\alpha$ either. To satisfy Theorem \ref{thm:reg-explicit},
we only require the coefficients $\xi_{i}^{k}$ be real. The SLD for
a pure state is $L_{i\bm{\lambda}}=2(\ket{\partial_{i}\psi_{\bm{\lambda}}}\bra{\psi_{\bm{\lambda}}}+\ket{\psi_{\bm{\lambda}}}\bra{\partial_{i}\psi_{\bm{\lambda}}})$
\citep{braunstein_statistical_1994,paris_quantum_2009}, from which
we find $L_{i\bm{\lambda}}^{\perp}=0$. Multiplying both sides of
Eq.~(\ref{eq:reg-explicit}) by $\braket{\pi_{k}\big|\psi_{\bm{\lambda}}}$,
the only requirement that $\xi_{i}^{k}$ be real gives 
\begin{equation}
\text{Im}[\braket{\partial_{i}^{0}\psi_{\bm{\lambda}}\big|\pi_{k}}\braket{\pi_{k}\big|\psi_{\bm{\lambda}}}]=0.
\end{equation}
This equation is equivalent as the Eq.~(8) in Pezz\`e et al. \citep{pezz`e2017optimal}
and is a generalization of Eq.~(29) in Braunstein and Caves \citep{braunstein_statistical_1994}.
Similarly, taking $\ket{\tilde{\psi}_{n\bm{\lambda}}}$ as $\ket{\psi_{\bm{\lambda}}}$
and multiplying both sides of Eq.~(\ref{eq:null-explicit}) by $\braket{\pi_{k}\big|\partial_{j}\psi_{\bm{\lambda}}}$,
the only requirement that $\eta_{ij}^{k}$ be real gives 
\begin{equation}
\text{Im}[\braket{\partial_{i}\psi_{\bm{\lambda}}\big|\pi_{k}}\braket{\pi_{k}\big|\partial_{j}\psi_{\bm{\lambda}}}]=0,
\end{equation}
which recovers Eq. (7) of Pezz\`e et al. \citep{pezz`e2017optimal}.

\section{Application to the three-dimensional imaging of two incoherent optical
point sources\label{sec:superresolution}}

Let us now apply the above theorems to estimate the three-dimensional
separation of two \textit{incoherent} point sources of monochromatic
light. Fig. \ref{fig:Setup} shows the basic setup of the problem:
The longitudinal axis ($Z$ axis) is taken to be the direction of
light propagation. We assume the coordinates of the centroid of the
two sources is known and chosen as the origin. The coordinates of
the two sources are $\pm\bm{s}\equiv\pm(s_{1},\,s_{2},\,s{}_{3})$
respectively. The transverse coordinates are denoted as $\bm{s}_{\perp}\equiv(s_{1},\,s_{2})$
and the dimensionless coordinates at the pupil plane are denoted as
$\bm{r}=(x_{1},\,x_{2})$ respectively. We consider the one photon
mixed state \textit{$\rho_{\bm{s}}=1/2\ket{\Psi_{+\bm{s}}}\bra{\Psi_{+\bm{s}}}+1/2\ket{\Psi_{-\bm{s}}}\bra{\Psi_{-\bm{s}}}$,
}where $\ket{\Psi_{\pm\bm{s}}}\equiv e^{i\theta_{\pm\bm{s}}}\ket{\Phi_{\pm\bm{s}}}$.
The pupil function is $\Phi_{\bm{s}}(\bm{r})\equiv\braket{\bm{r}\big|\Phi_{\bm{s}}}=\mathcal{A}\text{circ}(r/a)$$\allowbreak\exp[ik(\bm{s}_{\perp}\cdot\bm{r}-s_{3}r^{2}/2)]$
\citep{backlund2018fundamental,richards1959electromagnetic} , where
the normalization constant $\mathcal{A}=1/(\sqrt{\pi}a)$, $\text{circ}(r/a)$
is one if $0\le r\le a$ and vanishes everywhere else, and $r=\sqrt{x_{1}^{2}+x_{2}^{2}}$.
The overall phase $\theta_{\bm{s}}$ is chosen such that $\Delta_{\bm{s}}\equiv e^{2i\theta_{\bm{s}}}\int d\bm{r}\Phi_{\bm{s}}^{2}(\bm{r})$
is real. Due to Eq. (\ref{eq:theta-odd}), we find $\braket{\bm{r}\big|\Psi_{-\bm{s}}}\equiv e^{-i\theta_{\bm{s}}}\Phi_{-\bm{s}}(\bm{r})$
and $\braket{\Psi_{-\bm{s}}\big|\Psi_{+\bm{s}}}=\Delta_{\bm{s}}$
is also real. With this observation, we can diagonalize $\rho_{\bm{s}}$
with the states $\ket{\psi_{1\bm{s}}}=\ket{\tilde{\psi}_{1\bm{\lambda}}}/\sqrt{4p_{1\bm{s}}}$
and $\ket{\psi_{2\bm{s}}}=\ket{\tilde{\psi}_{2\bm{s}}}/\sqrt{4p_{2\bm{s}}}$,
where $\ket{\tilde{\psi}_{1\bm{s}}}=\ket{\Psi_{+\bm{s}}}+\ket{\Psi_{-\bm{s}}}$,
$\ket{\tilde{\psi}_{2\bm{s}}}=-\text{i}(\ket{\Psi_{+\bm{s}}}-\ket{\Psi_{-\bm{s}}})$
and the corresponding eigenvalues $p_{1,2\bm{s}}=(1\pm\Delta_{\bm{s}})/2$.
The QFIM associated with $\rho_{\bm{s}}$ has been shown in Ref. \citep{yu2018quantum},
which is $I_{ij}=4\text{Re}[\braket{\partial_{i}\Phi_{\bm{s}}\big|\partial_{j}\Phi_{\bm{s}}}+\allowbreak\braket{\Phi_{\bm{s}}\big|\partial_{i}\Phi_{\bm{s}}}\braket{\Phi_{\bm{s}}\big|\partial_{j}\Phi_{\bm{s}}}]$.
A straightforward calculation shows that the QFIM is diagonal with
diagonal matrix elements $k^{2}a^{2}$, $k^{2}a^{2}$ and $k^{2}a^{4}/12$.
We will focus on the saturation of the QFIM subsequently and construct
the corresponding optimal measurements.

Since now we have successfully diagonalized $\rho_{\bm{s}}$, we can
apply Theorems \ref{thm:reg-explicit}-\ref{thm:null-explicit} to
this problem to obtain the necessary and sufficient conditions for
optimal measurements. We summarize the results as two corollaries
below. The proofs can be found in Appendices \ref{subsec:s0} and
\ref{subsec:sperp0}. Note that our approach to optimal measurements
is quite different from the approach of direct calculations by many
papers \citep{tsang2016quantum,ang2017quantum,Rehacek-17-OL,yu2018quantum},
where one needs to calculate the QFIM first and then check whether
the CFIM associated with a specific measurement coincides with the
QFIM. 

\begin{cor}\label{cor:s-zero}The matrix bound of CFIM corresponding
to a projector $\Pi_{k}=\sum_{\alpha}\ket{\pi_{k\alpha}}\bra{\pi_{k\alpha}}$
can be saturated locally at in the limit $\bm{s}\to0$ \footnote{When $\bm{s}=0$, $\rho_{\bm{s}}$ becomes locally pure, which is
quite different from globally pure states discussed in \citep{pezz`e2017optimal}.
This is the critical point where the rank of $\rho_{\bm{s}}$ changes.
To find the saturation conditions, one should explore Eqs.~(\ref{eq:reg-explicit})
in the limit $\bm{s}\to0$. See Appendix \ref{subsec:s0} for the
proof.}, if and only if 
\begin{equation}
\braket{\partial_{i}^{0}\Phi_{\bm{s}}\big|\pi_{k\alpha}}\big|_{\bm{s}=0}=0\,\forall i,\,\alpha,
\end{equation}
provided the projector is \textit{regular} and if and only if 
\begin{equation}
\braket{\partial_{i}\Phi_{\bm{s}}\big|\pi_{k\alpha}}\big|_{\bm{s}=0}=\eta_{ij}^{k}\braket{\partial_{j}\Phi_{\bm{s}}\big|\pi_{k\alpha}}\big|_{\bm{s}=0}\,\forall i,\,j,\,\alpha,
\end{equation}
provided the projector is \textit{null}, where $\xi_{i}^{k}$ and
$\eta_{ij}^{k}$ are real and independent of $\alpha$.\end{cor}

\begin{cor}\label{cor:sperp-zero}On the line $\bm{s}_{\perp}=0$
the matrix bound of CFIM of estimating the \textit{transverse} separation
corresponding to a projector $\Pi_{k}=\sum_{\alpha}\ket{\pi_{k\alpha}}\bra{\pi_{k\alpha}}$
can be saturated locally, if and only if 
\begin{equation}
\braket{\partial_{i}\tilde{\psi}_{n\bm{s}}\big|\pi_{k\alpha}}=\xi_{i}^{k}\braket{\tilde{\psi}_{n\bm{s}}\big|\pi_{k\alpha}},\,i=1,\,2,\,\forall n,\,\alpha,\label{eq:diPsi-reg-Pi}
\end{equation}
provided the projector is \textit{regular}, and if and only if 
\begin{equation}
\braket{\partial_{i}\tilde{\psi}_{n\bm{s}}\big|\pi_{k\alpha}}=\eta_{ij}^{k}\braket{\partial_{j}\tilde{\psi}_{n\bm{s}}\big|\pi_{k\alpha}},\,i,\,j=1,\,2,\,\forall n,\,\alpha,\label{eq:diPsi-null-Pi}
\end{equation}
provided the projector is \textit{null}, where $\xi_{i}^{k}$ and
$\eta_{ij}^{k}$ are real and independent of $n$ and $\alpha$.\end{cor}

Based on Corollary \ref{cor:sperp-zero}, we propose the following
\textit{recipe} of searching for the optimal measurements: (i) Identify
the regular and null basis vectors in a given complete and orthonormal
basis $\{\ket{\pi_{k\alpha}}\}$. (ii) For each regular basis vector
$\ket{\pi_{k\alpha}}$, calculate the coefficient $\xi_{i}^{k\alpha}$
defined in Eq.~(\ref{eq:diPsi-reg-Pi}) and check whether $\xi_{i}^{k\alpha}$'s
are real for each $i$ and independent of the index $n$. (iii) Assemble
regular basis vectors that have the same coefficient $\xi_{i}^{k\alpha}$
as a regular projector $\Pi_{k}=\sum_{\alpha}\ket{\pi_{k\alpha}}\bra{\pi_{k\alpha}}$.
(iv) A null basis vector $\ket{\pi_{k\alpha}}$ is\textit{ flexible}
if $\braket{\partial_{i}\tilde{\psi}_{n\bm{s}}\big|\pi_{k\alpha}}=0$
for all $n$ and $i$. The rank one projector $\Pi_{k\alpha}$ formed
by a flexible basis vector can be added to any of the previous regular
projectors or the following null projectors. (v) For a null basis
vector that is not flexible, calculate the upper or lower triangular
(excluding diagonal) matrix elements $\eta_{ij}^{k\alpha}$ defined
in Eq.~(\ref{eq:diPsi-null-Pi}) and check whether they are all real
and independent of the index $n$. If for $n=1,\,2,$ both $\braket{\partial_{i}\tilde{\psi}_{n\bm{s}}\big|\pi_{k\alpha}}$
and $\braket{\partial_{j}\tilde{\psi}_{n\bm{s}}\big|\pi_{k\alpha}}$
vanishes for some $i$ and $j$, $\eta_{ij}^{k\alpha}$ can be set
arbitrarily. (vi) Assemble null basis vectors that have the same $\eta$
matrix as a null projector $\Pi_{k}=\ket{\pi_{k\alpha}}\bra{\pi_{k\alpha}}$.
A similar recipe can also be constructed based on Corollary \ref{cor:s-zero}.
It is clear from Theorems \ref{thm:reg-explicit}-\ref{thm:null-explicit}
that any partition of a set of optimal projectors is also optimal.
However, from the experimental point of view, we would like to minimize
the number of projectors for an optimal measurement. 

For the case of $\bm{s}=0$, we consider the Zernike basis vectors
denoted as $\ket{Z_{n}^{m}}$ \citep{born1999principles}, where $\ket{Z_{0}^{0}}=\ket{\Phi_{\bm{s}}}\big|_{\bm{s=0}}$.
Following the recipe above (details can be found in Appendix \ref{subsec:s0}):
(i) $\ket{Z_{0}^{0}}$ is the only regular basis vector and the remaining
basis vectors are null. (ii) We find $\braket{\partial_{i}^{0}\Phi_{\bm{s}}\big|Z_{0}^{0}}\big|_{\bm{s}=0}=0$
for $i=1,\,2,\,3$. (iii) Thus we obtain a regular projector $\ket{Z_{0}^{0}}\bra{Z_{0}^{0}}$.
(iv) For null basis vectors $\ket{Z_{n}^{m}}$ with $(n,\,m)\neq(1,\pm1),\,(2,\,0)$,
we find $\braket{\partial_{i}\Phi_{\bm{s}}\big|Z_{n}^{m}}=0$ for
all $i$. Thus these basis vectors are flexible and can be lumped
to the previous regular projector to form a new regular projector
$\Pi_{1}$. (v-vi) We then calculate the $\eta$ matrices corresponding
to the null basis vectors $\ket{Z_{1}^{\pm1}}$ and $\ket{Z_{2}^{0}}$
and find they are all distinct. Therefore we obtain three more null
projectors $\Pi_{2}=\ket{Z_{1}^{1}}\bra{Z_{1}^{1}}$, $\Pi_{3}=\ket{Z_{1}^{-1}}\bra{Z_{1}^{-1}}$,
and $\Pi_{4}=\ket{Z_{2}^{0}}\bra{Z_{2}^{0}}$. We conclude that the
projectors $\{\Pi_{k}\}_{k=1}^{4}$ are the optimal measurement in
the limit $\bm{s}\to0$.

\begin{figure}
\begin{centering}
\includegraphics[scale=0.45]{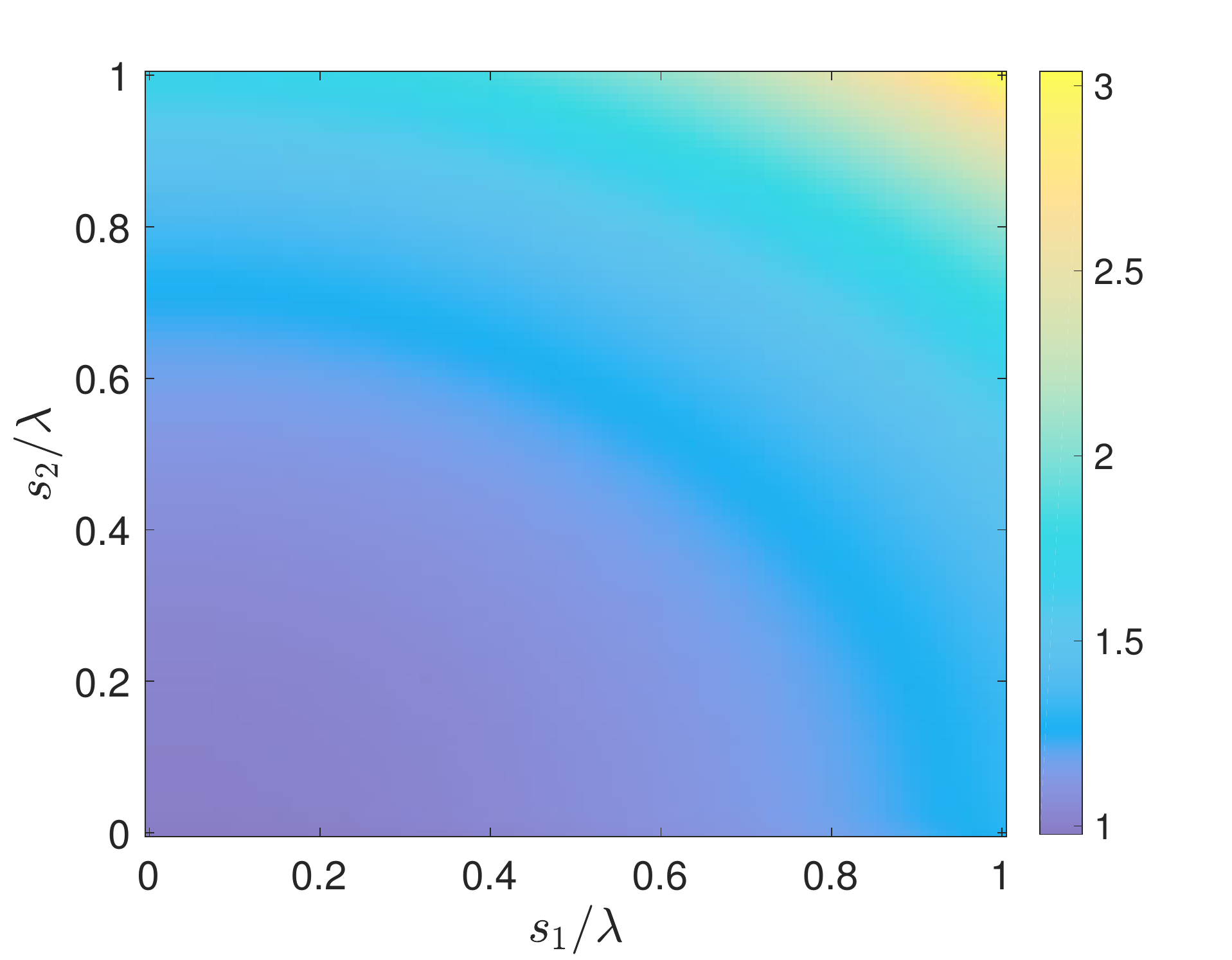}
\par\end{centering}
\centering{}\caption{\label{fig:Numerical-simulations}Numerical simulation of the classical
Cram\'er-Rao bound associated with the optimal measurement: $\Pi_{1}=\sum_{n=0}^{\infty}\ket{Z_{2n+1}^{1}}\bra{Z_{2n+1}^{1}}$,
$\Pi_{2}=\sum_{n=0}^{\infty}\ket{Z_{2n+1}^{-1}}\bra{Z_{2n+1}^{-1}}$
and $\Pi_{3}=1-\Pi_{1}-\Pi_{2}$. Note that the QFIM of estimating
the transverse separation is diagonal with both diagonal matrix elements
$k^{2}a^{2}$. The parameter setting is $a=0.2$, $\lambda=1$, $k=2\pi/\lambda$,
$s_{3}=5\lambda$. The plotted quantity is $k^{2}a^{2}(F^{-1})_{11}$,
where $F$ denotes the CFIM. As we can see, near the origin the quantum
Cram\'er-Rao bound of estimating $s_{1}$ is saturated. }
\end{figure}

On the line $\bm{s}_{\perp}=0$, we are interested in estimating the
transverse separation and therefore set $i=1,\,2$. After some algebra,
it is readily shown that $\tilde{\psi}_{1,2\bm{s}}(\bm{r})\big|_{\bm{s}_{\perp}=0}$
are even and $\partial_{i}\tilde{\psi}_{1,2\bm{s}}(\bm{r})\big|_{\bm{s}_{\perp}=0}$
are odd. We still consider Zernike basis vectors $\ket{Z_{n}^{m}}$.
Following the previously proposed recipe (details can be found in
Appendix \ref{subsec:sperp0}: (i) Even basis vectors are either regular
or flexible. Odd basis vectors are null and they are also flexible
except for $m=\pm1$. (ii) For regular and even basis vectors $\ket{Z_{2n}^{2m}}$,
it is easily calculated that $\xi_{i}^{(2n,2m)}=0$ for all $i$ and
$\ket{\tilde{\psi}_{1,\,2\bm{s}}}$. (iii) Thus we can construct a
regular projector as a sum of the rank one projectors formed by all
the regular and even basis vectors. (iv) We add all the rank one projectors
formed by flexible basis vectors to the previous regular projector
to obtain a regular projector $\Pi_{1}=1-\sum_{n=0}^{\infty}\ket{Z_{2n+1}^{\pm1}}\bra{Z_{2n+1}^{\pm1}}$.
(v) For the remaining null basis vectors, where $m=\pm1$, we find
$\eta_{21}^{(2n+1,\,1)}=0$ and $\eta_{12}^{(2n+1,\,-1)}=0$ for $\ket{\psi_{1,\,2\bm{s}}}$.
(vi) Since the set $\{\ket{Z_{2n+1}^{1}}\}$ has the same $\eta$
matrix and so does the set $\{\ket{Z_{2n+1}^{-1}}\}$, we obtain two
null projectors $\Pi_{2}=\sum_{n=0}^{\infty}\ket{Z_{2n+1}^{1}}\bra{Z_{2n+1}^{1}}$,
$\Pi_{3}=\sum_{n=0}^{\infty}\ket{Z_{2n+1}^{-1}}\bra{Z_{2n+1}^{-1}}$.
Note that these optimal projectors are independent of functional form
of the radial parts of the Zernike basis functions due to the fact
that the radial parts for a fixed angular index $m$ are complete
in the radial subspace. In fact, for a state $\braket{\bm{r}\big|\psi}=\psi(r,\,\phi)$,
one can show that $\braket{\psi\big|\Pi_{2}\big|\psi}\allowbreak$$=1/\pi\int_{0}^{\infty}rdr\big|\int_{0}^{2\pi}d\phi\psi(r,\,\phi)\cos\phi\big|^{2}$
and $\braket{\psi\big|\Pi_{3}\big|\psi}\allowbreak$$=1/\pi\int_{0}^{\infty}rdr\big|\int_{0}^{2\pi}d\phi\psi(r,\,\phi)\sin\phi\big|^{2}$,
where one can explicitly see that the probabilities do not depend
on the functional form of the radial parts of the basis functions.
Furthermore the probability distribution corresponding to such a measurement
is insensitive to the small change in the longitudinal separation.
Thus one cannot extract any information about $s_{3}$ from this measurement.
Fig.~\ref{fig:Numerical-simulations} is the numerical calculation
of classical CR bound of estimating $s_{1}$ associated with this
measurement. As we clearly see from Fig.~\ref{fig:Numerical-simulations},
the Helstrom CR bound of estimating $s_{1}$ is saturated near the
origin where $\bm{s}_{\perp}=0$. Note that the Helstrom CR bound
of estimating $s_{2}$ is the same as that of estimating $s_{1}$
and hence is omitted here.

On the plane $s_{3}=0$, for the case of $i=1,\,2$, i.e., estimating
the transverse separation, following the previous recipe (see Appendix
\ref{subsec:s30} for details) we find rank one projectors formed
by real and parity definite basis functions are optimal on the plane
$s_{3}=0$. This result is a generalization of previous one-dimensional
transverse estimation \citep{Rehacek-17-OL}.

\section{Conclusion\label{sec:Conclusion}}

We gave the necessary and sufficient conditions for any POVM measurement
to give the Helstrom CR bound. Based on these saturation conditions,
we predicted several local optimal measurements in the problem of
estimating the three-dimensional separation of two incoherent light
sources. These predictions are confirmed by numerical simulations.

Based on our results here, many open questions can be further explored,
such as searching for a general recipe for the optimal measurement
common to all parameters when the partial commutativity condition
is satisfied, saturating the QFIM asymptotically due to collective
measurements on a large number of identical states, etc. Our work
has potential applications in quantum sensing, quantum enhanced imaging,
in particular may shed light on investigating the attainability of
the Helstrom CR bound for an initial probe state undergoing noisy
dynamics and moment estimation in quantum imaging of finite number
of point sources.

\section{Acknowledgments}

We thank Marco Genoni, Haidong Yuan, Sisi Zhou and Liang Jiang for
useful discussions. This work was supported by U.S. Army Research
Office Grants No. W911NF-18-10178, No. W911NF-15-1-0496, No. W911NF-13-1-0402,
and by National Science Foundation Grants No. DMR-1809343, No. DMR-1506081.
Y. Z. is supported by U.S. Office of Naval Research.

\appendix

\section{\label{sec:Matrix-SLD}The matrix representation of the SLD}

We denote the orthonormal basis vectors of the support and the kernel
of $\rho_{\bm{\lambda}}$ as $\ket{\psi_{n\bm{\lambda}}}$ and $\ket{e_{n\bm{\lambda}}}$
respectively. Then Eq.~(\ref{eq:SLD}), the defining equation of
the SLD, in the basis formed by $\{\ket{\psi_{n\bm{\lambda}}},\,\ket{e_{n\bm{\lambda}}}\}$
becomes 
\begin{equation}
[L_{i\bm{\lambda}}]_{mn}\varrho_{n\bm{\lambda}}+[L_{i\bm{\lambda}}]_{mn}\varrho_{m\bm{\lambda}}=2[\partial_{i}\rho_{\bm{\lambda}}]_{mn},\label{eq:MatrixDefSLD}
\end{equation}
where $[L_{i\bm{\lambda}}]_{mn}\equiv\braket{m\big|L_{i\bm{\lambda}}\big|n}$,
$\ket{n}$ is the eigenvector of $\rho_{\bm{\lambda}}$ which could
be either $\ket{\psi_{n\bm{\lambda}}}$ or $\ket{e_{n\bm{\lambda}}}$,
\begin{align}
\partial_{i}\rho_{\bm{\lambda}} & =\sum_{n}\partial_{i}p_{n\bm{\lambda}}\ket{\psi_{n\bm{\lambda}}}\bra{\psi_{n\bm{\lambda}}}+\sum_{n}p_{n\bm{\lambda}}\ket{\partial_{i}\psi_{n\bm{\lambda}}}\bra{\psi_{n\bm{\lambda}}}\nonumber \\
 & +\sum_{n}p_{n\bm{\lambda}}\ket{\psi_{n\bm{\lambda}}}\bra{\partial_{i}\psi_{n\bm{\lambda}}}\label{eq:dirho}
\end{align}
and $\varrho_{n\bm{\lambda}}$ is the corresponding eigenvalue which
could be either the positive eigenvalue $p_{n\bm{\lambda}}$ or zero.
From Eq.~(\ref{eq:dirho}), we know that for $\ket{m}=\ket{e_{m\bm{\lambda}}}$
and $\ket{n}=\ket{e_{n\bm{\lambda}}}$, where $\varrho_{m}=\varrho_{n}=0$,
\begin{equation}
\braket{e_{m\bm{\lambda}}\big|\partial_{i}\rho_{\bm{\lambda}}\big|e_{n\bm{\lambda}}}=0,\,\forall m,\,n.
\end{equation}
Thus we can choose 
\begin{equation}
\braket{e_{m\bm{\lambda}}\big|L_{i\bm{\lambda}}\big|e_{n\bm{\lambda}}}=0.
\end{equation}
Therefore the following choice of the SLD
\begin{equation}
[L_{i\bm{\lambda}}]_{mn}=\begin{cases}
0 & \ket{m}=\ket{e_{m\bm{\lambda}}}\:\text{and}\:\ket{n}=\ket{e_{n\bm{\lambda}}}\\
\frac{2[\partial_{i}\rho_{\bm{\lambda}}]_{mn}}{\varrho_{m\bm{\lambda}}+\varrho_{n\bm{\lambda}}} & \text{else}
\end{cases}\label{eq:Limn}
\end{equation}
can satisfy its matrix definition Eq.~(\ref{eq:MatrixDefSLD}). Based
on Eqs.~(\ref{eq:dirho}, \ref{eq:Limn}), a matrix representation
of the SLD, 
\begin{align}
L_{i\bm{\lambda}} & =\sum_{n}\frac{\partial_{i}p_{n,\,\bm{\lambda}}}{p_{n,\,\bm{\lambda}}}\ket{\psi_{n\bm{\lambda}}}\bra{\psi_{n\bm{\lambda}}}\nonumber \\
 & +2\sum_{m,\,n}\frac{p_{m\bm{\lambda}}-p_{n\bm{\lambda}}}{p_{m\bm{\lambda}}+p_{n\bm{\lambda}}}\braket{\partial_{i}\psi_{m\bm{\lambda}}\big|\psi_{n\bm{\lambda}}}\ket{\psi_{m\bm{\lambda}}}\bra{\psi_{n\bm{\lambda}}}\nonumber \\
 & +\left[2\sum_{m,\,n}\braket{\partial_{i}\psi_{n\bm{\lambda}}\big|e_{m\bm{\lambda}}}\ket{\psi_{n\bm{\lambda}}}\bra{e_{m\bm{\lambda}}}+\text{c.c.}\right],\label{eq:Matrix-SLD}
\end{align}
can be found \citep{paris_quantum_2009}. With Eq.~(\ref{eq:Matrix-SLD}),
it is straightforward to calculate, for a state $\ket{\psi}=\ket{\psi^{0}}+\ket{\psi^{\perp}}$,
\begin{align}
L_{i\bm{\lambda}}\ket{\psi} & =L_{i\bm{\lambda}}^{\perp}\ket{\psi}+2\sum_{m,\,n}\braket{\partial_{i}\psi_{n\bm{\lambda}}\big|e_{m\bm{\lambda}}}\braket{e_{m\bm{\lambda}}\big|\psi^{0}}\ket{\psi_{n\bm{\lambda}}}\nonumber \\
 & +2\sum_{m,\,n}\braket{e_{m\bm{\lambda}}\big|\partial_{i}\psi_{n\bm{\lambda}}}\braket{\psi_{n,\,\bm{\lambda}}\big|\psi^{\perp}}\ket{e_{m,\,\bm{\lambda}}},\label{eq:LiPsi}
\end{align}
where 
\begin{align}
L_{i\bm{\lambda}}^{\perp} & \equiv\sum_{n}\frac{\partial_{i}p_{n\bm{\lambda}}}{p_{n\bm{\lambda}}}\ket{\psi_{n\bm{\lambda}}}\bra{\psi_{n\bm{\lambda}}}\nonumber \\
 & +2\sum_{m,n}\frac{p_{m\bm{\lambda}}-p_{n\bm{\lambda}}}{p_{m\bm{\lambda}}+p_{n\bm{\lambda}}}\braket{\partial_{i}\psi_{m\bm{\lambda}}\big|\psi_{n\bm{\lambda}}}\ket{\psi_{m\bm{\lambda}}}\bra{\psi_{n\bm{\lambda}}},\label{eq:Liperp}
\end{align}
is the projection of the SLD on the subspace $\text{supp}(\rho_{\bm{\lambda}})$.
Upon noting the following identities
\begin{equation}
\braket{\partial_{i}\psi_{n\bm{\lambda}}\big|e_{m\bm{\lambda}}}=\braket{\partial_{i}^{0}\psi_{n\bm{\lambda}}\big|e_{m\bm{\lambda}}},
\end{equation}
\begin{align}
\sum_{m}\braket{\partial_{i}^{0}\psi_{n\bm{\lambda}}\big|e_{m\bm{\lambda}}}\braket{e_{m\bm{\lambda}}\big|\psi^{0}}\nonumber \\
=\braket{\partial_{i}^{0}\psi_{n\bm{\lambda}}\big|\psi^{0}}=\braket{\partial_{i}^{0}\psi_{n\bm{\lambda}}\big|\psi},
\end{align}
we obtain
\begin{equation}
\braket{\psi_{n\bm{\lambda}}\big|L_{i\bm{\lambda}}\big|\psi}=\braket{\psi_{n\bm{\lambda}}\big|L_{i\bm{\lambda}}^{\perp}\big|\psi}+2\braket{\partial_{i}^{0}\psi_{n\bm{\lambda}}\big|\psi}.\label{eq:LiPsiPerp}
\end{equation}

\section{Details on the example of estimating the separations of two incoherent
optical point sources}

\subsection{Properties of $\ket{\psi_{1,\,2\bm{s}}}$\label{subsec:psi12Properties}}

We mention in the main text that $\theta_{\bm{s}}$ is chosen such
that 
\begin{equation}
\Delta_{\bm{s}}\equiv e^{2\text{i}\theta_{\bm{s}}}\int d\bm{r}\Phi_{\bm{s}}^{2}(\bm{r})\label{eq:Delta-def}
\end{equation}
 is real. Defining 
\begin{align}
v_{\bm{s}} & \equiv\text{Re}[\int d\bm{r}\Phi_{\bm{s}}^{2}(\bm{r})]\nonumber \\
 & =\mathcal{A}^{2}\int d\bm{r}\text{circ}(r/a)\cos(2k\bm{s}_{\perp}\cdot\bm{r})\cos(ks_{3}r^{2}),\label{eq:vs}
\end{align}
\begin{align}
w_{\bm{s}} & \equiv\text{Im}[\int d\bm{r}\Phi_{\bm{s}}^{2}(\bm{r})]\nonumber \\
 & =-\mathcal{A}^{2}\int d\bm{r}\text{circ}(r/a)\cos(2k\bm{s}_{\perp}\cdot\bm{r})\sin(ks_{3}r^{2}),\label{eq:ws}
\end{align}
we can express $\theta_{\bm{s}}$ and $\Delta_{\bm{s}}$ as 
\begin{equation}
\tan2\theta_{\bm{s}}=-\frac{w_{\bm{s}}}{v_{\bm{s}}},\label{eq:tan2theta}
\end{equation}
\begin{equation}
\Delta_{\bm{s}}=\sqrt{v_{\bm{s}}^{2}+w_{\bm{s}}^{2}}.\label{eq:delta}
\end{equation}
As is clear from Eqs.~(\ref{eq:vs}, \ref{eq:ws}), $v_{\bm{s}}$
is even in $\bm{s}$ while $w_{\bm{s}}$ is odd in $\bm{s}$. Thus
according to Eq.~(\ref{eq:tan2theta}), we know $\theta_{\bm{s}}$
is odd in $\bm{s}$, i.e., 
\begin{equation}
\theta_{\bm{s}}=-\theta_{-\bm{s}}.\label{eq:theta-odd}
\end{equation}
With these observations, the one photon state defined in the main
text can be diagonalized by the following state:
\begin{eqnarray}
\psi_{1\bm{s}}(\bm{r}) & = & \frac{1}{\sqrt{2(1+\Delta_{\bm{s}})}}[\Psi_{+\bm{s}}(\bm{r})+\Psi_{-\bm{s}}(\bm{r})]\nonumber \\
 &  & =\frac{\tilde{\psi}_{1\bm{s}}(\bm{r})}{\sqrt{4p_{1\bm{s}}}},\label{eq:psi1r}\\
\psi_{2\bm{s}}(\bm{r}) & = & \frac{-\text{i}}{\sqrt{2(1-\Delta_{\bm{s}})}}[\Psi_{+\bm{s}}(\bm{r})-\Psi_{-\bm{s}}(\bm{r})]\nonumber \\
 &  & =\frac{\tilde{\psi}_{2\bm{s}}(\bm{r})}{\sqrt{4p_{2\bm{s}}}}\label{eq:psi2r}
\end{eqnarray}
where 
\begin{equation}
\Psi_{\pm\bm{s}}(\bm{r})=e^{\pm i\theta_{\bm{s}}}\Phi_{\pm\bm{s}}(\bm{r}),\label{eq:PsiPM-Phis}
\end{equation}
\begin{equation}
\Phi_{\bm{s}}(\bm{r})=\mathcal{A}\text{circ}(r/a)\exp[ik(\bm{s}_{\perp}\cdot\bm{r}-s_{3}r^{2}/2)].\label{eq:Phis}
\end{equation}
We can write the explicit forms of $\psi_{\pm\bm{s}}(\bm{r})\equiv\braket{\bm{r}\big|\psi_{\pm\bm{s}}}$
as
\begin{align}
\psi_{1\bm{s}}(\bm{r}) & =\frac{\tilde{\psi}_{1\bm{s}}(\bm{r})}{\sqrt{4p_{1\bm{s}}}}=\frac{2\mathcal{A}\text{circ}(r/a)}{\sqrt{4p_{1\bm{s}}}},\nonumber \\
 & \times\cos(\theta_{\bm{s}}+k\bm{s}_{\perp}\cdot\bm{r}-ks_{3}r^{2}/2)\label{eq:psiP-expression}
\end{align}
\begin{align}
\psi_{2\bm{s}}(\bm{r}) & =\frac{\tilde{\psi}_{2\bm{s}}(\bm{r})}{\sqrt{4p_{2\bm{s}}}}=\frac{2\mathcal{A}\text{circ}(r/a)}{\sqrt{4p_{2\bm{s}}}}.\nonumber \\
 & \times\sin(\theta_{\bm{s}}+k\bm{s}_{\perp}\cdot\bm{r}-ks_{3}r^{2}/2)\label{eq:psiM-expression}
\end{align}
Eqs. (\ref{eq:PsiPM-Phis}, \ref{eq:Phis}) immediately tell us 
\begin{align}
\braket{\partial_{i}\Psi_{+\bm{s}}\big|\Psi_{+\bm{s}}} & =-\braket{\partial_{i}\Psi_{-\bm{s}}\big|\Psi_{-\bm{s}}}\nonumber \\
 & =-i\partial_{i}\theta_{\bm{s}}+i\delta_{i3}ka^{2}/4,\label{eq:diPsiPP}
\end{align}
\begin{align}
\braket{\partial_{i}\Psi_{-\bm{s}}\big|\Psi_{+\bm{s}}} & =\braket{\partial_{i}\Psi_{+\bm{s}}\big|\Psi_{-\bm{s}}}\nonumber \\
 & =-\partial_{i}\Delta_{\bm{s}}/2,\label{eq:diPsiMP}
\end{align}
where $\delta_{i3}$ is the Kronecker delta.

From Eqs.~(\ref{eq:psi1r}, \ref{eq:psi2r}) we know that $\psi_{1,\,2\bm{s}}(\bm{r})$
are real. Therefore we conclude $\braket{\partial_{i}\psi_{1\bm{s}}\big|\psi_{1\bm{s}}}$
and $\braket{\partial_{i}\psi_{2\bm{s}}\big|\psi_{2\bm{s}}}$ must
be real. On other hand, they must be purely imaginary due to the fact
that $\braket{\partial_{i}\psi_{n\bm{s}}\big|\psi_{n\bm{s}}}+\braket{\psi_{n\bm{s}}\big|\partial_{i}\psi_{n\bm{s}}}=0$
for $n=1,\,2$. So we end up with
\begin{equation}
\braket{\partial_{i}\psi_{1\bm{s}}\big|\psi_{1\bm{s}}}=\braket{\partial_{i}\psi_{2\bm{s}}\big|\psi_{2\bm{s}}}=0.\label{eq:dipsiPpsiP}
\end{equation}
Furthermore $\braket{\partial_{i}\psi_{1\bm{s}}\big|\psi_{2\bm{s}}}$
is also real since, upon application of Eqs.~(\ref{eq:diPsiPP},
\ref{eq:diPsiMP}),
\begin{align}
\braket{\partial_{i}\psi_{1\bm{s}}\big|\psi_{2\bm{s}}} & =\frac{-\text{i}}{2\sqrt{1-\Delta_{\bm{s}}^{2}}}(\bra{\partial_{i}\Psi_{+\bm{s}}}+\bra{\partial_{i}\Psi_{-\bm{s}}})\nonumber \\
 & \times(\ket{\Psi_{+\bm{s}}}-\ket{\Psi_{-\bm{s}}})\nonumber \\
 & =\frac{-\text{i}}{2\sqrt{1-\Delta_{\bm{s}}^{2}}}(\braket{\partial_{i}\Psi_{+\bm{s}}\big|\Psi_{+\bm{s}}}-\braket{\partial_{i}\Psi_{-\bm{s}}\big|\Psi_{-\bm{s}}}\nonumber \\
 & +\cancel{\braket{\partial_{i}\Psi_{-\bm{s}}\big|\Psi_{+\bm{s}}}-\braket{\partial_{i}\Psi_{+\bm{s}}\big|\Psi_{-\bm{s}}}})\nonumber \\
 & =\frac{-\text{i}\braket{\partial_{i}\Psi_{+\bm{s}}\big|\Psi_{+\bm{s}}}}{\sqrt{1-\Delta_{\bm{s}}^{2}}}\nonumber \\
 & =\frac{-\partial_{i}\theta_{\bm{s}}+\delta_{i3}ka^{2}/4}{\sqrt{1-\Delta_{\bm{s}}^{2}}}.\label{eq:diPsi1Psi2}
\end{align}
The fact that $\partial_{i}\braket{\psi_{2\bm{s}}\big|\psi_{1\bm{s}}}=0$
gives $\braket{\partial_{i}\psi_{2\bm{s}}\big|\psi_{1\bm{s}}}=-\braket{\psi_{2\bm{s}}\big|\partial_{i}\psi_{1\bm{s}}}$.
On the other hand Eq. (\ref{eq:diPsi1Psi2}) tells us $\braket{\psi_{2\bm{s}}\big|\partial_{i}\psi_{1\bm{s}}}=\braket{\partial_{i}\psi_{1\bm{s}}\big|\psi_{2\bm{s}}}$.
Thus we know
\begin{equation}
\braket{\partial_{i}\psi_{2\bm{s}}\big|\psi_{1\bm{s}}}=-\braket{\partial_{i}\psi_{1\bm{s}}\big|\psi_{2\bm{s}}}.\label{eq:diPsi2Psi1}
\end{equation}

\subsection{\label{subsec:s0}The case $\bm{s}=0$}

We first apply Theorem 3 in the main text to obtain the following
Lemmas, which will be useful subsequently.

\begin{lem}For the mixed state $\rho_{\bm{s}}$, the matrix bound
of the CFIM due to a \textit{regular} projector $\Pi_{k}=\sum_{\alpha}\ket{\pi_{k\alpha}}\bra{\pi_{k\alpha}}$
is saturated if and only if 

\begin{gather}
\frac{\partial_{i}p_{1\bm{s}}}{p_{1\bm{s}}}\braket{\psi_{1\bm{s}}\big|\pi_{k\alpha}}-4p_{2\bm{s}}\braket{\partial_{_{i}}\psi_{1\bm{s}}\big|\psi_{2\bm{s}}}\braket{\psi_{2\bm{s}}\big|\pi_{k\alpha}}\nonumber \\
+2\braket{\partial_{i}\psi_{1\bm{s}}\big|\pi_{k\alpha}}=\xi_{i}^{k}\braket{\psi_{1\bm{s}}\big|\pi_{k\alpha}},\label{eq:Lemm1-cond1}
\end{gather}
\begin{gather}
\frac{\partial_{_{i}}p_{2\bm{s}}}{p_{2\bm{s}}}\braket{\psi_{2\bm{s}}\big|\pi_{k\alpha}}-4p_{1\bm{s}}\braket{\partial_{_{i}}\psi_{2\bm{s}}\big|\psi_{1\bm{s}}}\braket{\psi_{1\bm{s}}\big|\pi_{k\alpha}}\nonumber \\
+2\braket{\partial_{i}\psi_{2\bm{s}}\big|\pi_{k\alpha}}=\xi_{i}^{k}\braket{\psi_{2\bm{s}}\big|\pi_{k\alpha}},\label{eq:Lemma1-cond2}
\end{gather}
holds $\forall i,\,\alpha$, where $p_{1,\,2\bm{s}}=(1\pm\Delta_{\bm{s}})/2$
and $\xi_{i}^{k}$ is real and independent of $n$ and $\alpha$.
\end{lem}

\begin{proof}According to Eqs. (\ref{eq:Liperp}, \ref{eq:dipsiPpsiP}),
we find 
\begin{align}
L_{i\bm{s}}^{\perp} & =\frac{\partial_{i}\Delta_{\bm{s}}}{1+\Delta{}_{\bm{s}}}\ket{\psi_{1\bm{s}}}\bra{\psi_{1\bm{s}}}-\frac{\partial_{_{i}}\Delta_{\bm{s}}}{1-\Delta_{\bm{s}}}\ket{\psi_{2\bm{s}}}\bra{\psi_{2\bm{s}}}\nonumber \\
 & +2\Delta_{\bm{s}}\braket{\partial_{_{i}}\psi_{1\bm{s}}\big|\psi_{2\bm{s}}}\ket{\psi_{1\bm{s}}}\bra{\psi_{2\bm{s}}}\nonumber \\
 & -2\Delta_{\bm{s}}\braket{\partial_{_{i}}\psi_{2\bm{s}}\big|\psi_{1\bm{s}}}\ket{\psi_{2\bm{s}}}\bra{\psi_{1\bm{s}}}.
\end{align}
Thus \begin{widetext}
\begin{gather}
L_{i\bm{s}}^{\perp}\ket{\pi_{k\alpha}}=\ket{\psi_{1\bm{s}}}\left[\frac{\partial_{i}\Delta{}_{\bm{s}}}{1+\Delta{}_{\bm{s}}}\braket{\psi_{1\bm{s}}\big|\pi_{k\alpha}}+2\Delta_{\bm{s}}\braket{\partial_{_{i}}\psi_{1\bm{s}}\big|\psi_{2\bm{s}}}\braket{\psi_{2\bm{s}}\big|\pi_{k\alpha}}\right]\nonumber \\
+\ket{\psi_{2\bm{s}}}\left[-\frac{\partial_{_{i}}\Delta{}_{\bm{s}}}{1-\Delta_{\bm{s}}}\braket{\psi_{2\bm{s}}\big|\pi_{k\alpha}}-2\Delta{}_{\bm{s}}\braket{\partial_{_{i}}\psi_{2\bm{s}}\big|\psi_{1\bm{s}}}\braket{\psi_{1\bm{s}}\big|\pi_{k\alpha}}\right],
\end{gather}
\end{widetext}which can be rewritten as upon noting Eq. (\ref{eq:dipsiPpsiP}),
\begin{gather}
L_{i\bm{s}}^{\perp}\ket{\pi_{k\alpha}}=\ket{\psi_{1\bm{s}}}\left[\frac{\partial_{i}\Delta_{\bm{s}}}{1+\Delta{}_{\bm{s}}}\braket{\psi_{1\bm{s}}\big|\pi_{k\alpha}}+2\Delta{}_{\bm{s}}\braket{\partial_{_{i}}^{\perp}\psi_{1\bm{s}}\big|\pi_{k\alpha}}\right]\nonumber \\
+\ket{\psi_{2\bm{s}}}\left[-\frac{\partial_{_{i}}\Delta{}_{\bm{s}}}{1-\Delta{}_{\bm{s}}}\braket{\psi_{2\bm{s}}\big|\pi_{k\alpha}}-2\Delta{}_{\bm{s}}\braket{\partial_{_{i}}^{\perp}\psi_{2\bm{s}}\big|\pi_{k\alpha}}\right].
\end{gather}
In order to saturate the Helstrom CR bound, according to Theorem \ref{thm:reg-explicit}
in the main text, every regular projector $\Pi_{k}=\sum_{\alpha}\ket{\pi_{k\alpha}}\bra{\pi_{k\alpha}}$
must satisfy
\begin{gather}
\frac{\partial_{i}\Delta{}_{\bm{s}}}{1+\Delta{}_{\bm{s}}}\braket{\psi_{1\bm{s}}\big|\pi_{k\alpha}}+2\Delta{}_{\bm{s}}\braket{\partial_{_{i}}^{\perp}\psi_{1\bm{s}}\big|\pi_{k\alpha}}\nonumber \\
+2\braket{\partial_{i}^{0}\psi_{1\bm{s}}\big|\pi_{k\alpha}}=\xi_{i}^{k}\braket{\psi_{1\bm{s}}\big|\pi_{k\alpha}},
\end{gather}
\begin{gather}
-\frac{\partial_{_{i}}\Delta{}_{\bm{s}}}{1-\Delta{}_{\bm{s}}}\braket{\psi_{2\bm{s}}\big|\pi_{k\alpha}}-2\Delta{}_{\bm{s}}\braket{\partial_{_{i}}^{\perp}\psi_{2\bm{s}}\big|\pi_{k\alpha}}\nonumber \\
+2\braket{\partial_{i}^{0}\psi_{2\bm{s}}\big|\pi_{k\alpha}}=\xi_{i}^{k}\braket{\psi_{2\bm{s}}\big|\pi_{k\alpha}},
\end{gather}
where $\xi_{i}^{k}$ is real. With the facts that $\braket{\partial_{_{i}}^{\perp}\psi_{1,\,2\bm{s}}\big|\pi_{k\alpha}}+\braket{\partial_{i}^{0}\psi_{1,\,2\bm{s}}\big|\pi_{k\alpha}}=\braket{\partial_{i}\psi_{1,\,2\bm{s}}\big|\pi_{k\alpha}}$
and $p_{1,\,2\bm{s}}=(1\pm\Delta{}_{\bm{s}})/2$ , one can easily
concludes the proof. \end{proof}

\subsubsection{Proof of Corollary 1 in the main text}

\begin{proof}We find at $\bm{s}=0$ 
\begin{equation}
\Delta_{\bm{s}}=1
\end{equation}
and $p_{1\bm{s}}\big|_{\bm{s}=0}=1$ and $p_{2\bm{s}}\big|_{\bm{s}=0}=0$.
Therefore $p_{1\bm{s}}$ and $p_{2\bm{s}}$ attain their local maximum
and minimum respectively at $\bm{s}=0$, i.e., $\partial_{i}p_{1\bm{s}}=\partial_{i}p_{2\bm{s}}=0$,
which indicates 
\begin{equation}
\partial_{i}\Delta_{\bm{s}}=0,\,i=1,\,2,\,3.
\end{equation}
We recognize that $\bm{s}=0$ is the critical point of the change
of the rank of $\rho_{\bm{s}}$. In this case, both the normalized
vector $\ket{\psi_{2\bm{s}}}$ and the first term on the left hand
side of Eq.~(\ref{eq:Lemma1-cond2}) is not well-defined. Therefore
Eqs.~(\ref{eq:Lemm1-cond1}, \ref{eq:Lemma1-cond2}) should be understood
in the sense of the limit $\bm{s}\to0$. It can be also shown that
for $i=1,\,2,\,3$,
\begin{equation}
\partial_{i}\tilde{\psi}_{1\bm{s}}(\bm{r})\big|_{\bm{s}=0}=\partial_{i}\psi_{1\bm{s}}(\bm{r})\big|_{\bm{s}=0}=0,\label{eq:diPsi-r-sZero}
\end{equation}
\begin{equation}
\ket{\tilde{\psi}_{2\bm{s}}}\big|_{\bm{s}=0}=0
\end{equation}
With these facts, the second term of the left hand side of Eq. (\ref{eq:Lemm1-cond1})
reads
\begin{gather}
4p_{2\bm{s}}\braket{\partial_{_{i}}\psi_{1\bm{s}}\big|\psi_{2\bm{s}}}\braket{\psi_{2\bm{s}}\big|\pi_{k\alpha}}\nonumber \\
=\braket{\partial_{_{i}}\psi_{1\bm{s}}\big|\tilde{\psi}_{2\bm{s}}}\braket{\tilde{\psi}_{2\bm{s}}\big|\pi_{k\alpha}}=0,
\end{gather}
 which immediately tells us that for a regular projector $\Pi_{k}$
that saturates its matrix bound $\xi_{i}^{k}=0\:\forall i=1,\,2\,\,3.$
In order to saturate the matrix bound, it remains to show Eq. (\ref{eq:Lemma1-cond2})
is consistent with the result $\xi_{i}^{k}=0$. It is readily checked
that 
\begin{equation}
\braket{\partial_{i}\psi_{2\bm{s}}\big|\pi_{k\alpha}}=\frac{\braket{\partial_{i}\tilde{\psi}_{2\bm{s}}\big|\pi_{k\alpha}}}{\sqrt{4p_{2\bm{s}}}}-\frac{\partial_{i}p_{2\bm{s}}}{2p_{2\bm{s}}}\frac{\braket{\tilde{\psi}_{2\bm{s}}\big|\pi_{k\alpha}}}{\sqrt{4p_{2\bm{s}}}}
\end{equation}
\begin{gather}
\lim_{\bm{s}\to0}p_{1\bm{s}}\braket{\partial_{_{i}}\psi_{2\bm{s}}\big|\psi_{1\bm{s}}}\braket{\psi_{1\bm{s}}\big|\pi_{k\alpha}}\nonumber \\
=-\lim_{\bm{s}\to0}\frac{\braket{\tilde{\psi}_{2\bm{s}}\big|\partial_{_{i}}\psi_{1\bm{s}}}}{\sqrt{4p_{2\bm{s}}}}\braket{\psi_{1\bm{s}}\big|\pi_{k\alpha}}
\end{gather}
the left hand and right hand sides of Eq. (\ref{eq:Lemma1-cond2})
can be written as 
\begin{equation}
\text{LHS}=\lim_{\bm{s}\to0}\frac{2\braket{\tilde{\psi}_{2\bm{s}}\big|\partial_{_{i}}\psi_{1\bm{s}}}}{\sqrt{p_{2\bm{s}}}}\braket{\psi_{1\bm{s}}\big|\pi_{k\alpha}}+\frac{\braket{\partial_{i}\tilde{\psi}_{2\bm{s}}\big|\pi_{k\alpha}}}{\sqrt{p_{2\bm{s}}}}\label{eq:LHS}
\end{equation}
\begin{equation}
\text{RHS}=\xi_{i}^{k}\lim_{\bm{s}\to0}\braket{\psi_{2\bm{s}}\big|\pi_{k\alpha}}=\xi_{i}^{k}\lim_{\bm{s}\to0}\frac{\braket{\tilde{\psi}_{2\bm{s}}\big|\pi_{k\alpha}}}{\sqrt{4p_{2\bm{s}}}}\label{eq:RHS}
\end{equation}
Upon eliminating the factor $1/\sqrt{p_{2\bm{s}}}$ in both equations
and noting that $\lim_{\bm{s}\to0}\braket{\tilde{\psi}_{2\bm{s}}\big|\partial_{_{i}}\psi_{1\bm{s}}}=\lim_{\bm{s}\to0}\braket{\tilde{\psi}_{2\bm{s}}\big|\pi_{k\alpha}}=0$,
we find that $\text{LHS}=\text{RHS}$ is the consistent with the result
$\xi_{i}^{k}=0$ if and only if 
\begin{equation}
\lim_{\bm{s}\to0}\braket{\partial_{i}\tilde{\psi}_{2\bm{s}}\big|\pi_{k\alpha}}=\braket{\partial_{i}\tilde{\psi}_{2\bm{s}}\big|\pi_{k\alpha}}\big|_{\bm{s}=0}=0,\,\forall i,\,\alpha.\label{eq:diPsi2tildPikalpha-reg}
\end{equation}
It is straightforward to calculate for $i=1,\,2,\,3$, we have

\begin{align}
\partial_{i}\tilde{\psi}_{2\bm{s}}(\bm{r})\big|_{\bm{s}=0} & =-\text{i}[\partial_{i}\Psi_{+\bm{s}}(\bm{r})-\partial_{i}\Psi_{-\bm{s}}(\bm{r})]\big|_{\bm{s}=0}\nonumber \\
 & =[2\partial_{i}\theta_{\bm{s}}\Phi_{\bm{s}}(\bm{r})-2\text{i}\partial_{i}\Phi_{\bm{s}}(\bm{r})]\big|_{\bm{s}=0}\label{eq:diPsi2tild-r}
\end{align}
On the other hand, at $\bm{s}=0$, according to Eqs.~(\ref{eq:vs},
\ref{eq:ws}), the explicit forms of $v_{\bm{s}}$ and $w_{\bm{s}}$
can be expressed as
\begin{equation}
v_{\bm{s}}\big|_{\bm{s}=0}=\pi\mathcal{A}^{2}a^{2},
\end{equation}
\begin{equation}
w_{\bm{s}}\big|_{\bm{s}=0}=0,
\end{equation}
According to Eq.~(\ref{eq:tan2theta}), we obtain
\begin{equation}
\theta_{\bm{s}}\big|_{\bm{s}=0}=0.
\end{equation}
Differentiating both sides of Eq. (\ref{eq:Delta-def}), we obtain
\begin{equation}
\partial_{i}\Delta_{\bm{s}}=2\text{i}\partial_{i}\theta_{\bm{s}}\int d\bm{r}\Phi_{\bm{s}}^{2}(\bm{r})+2e^{2\text{i}\theta_{\bm{s}}}\int d\bm{r}\Phi_{\bm{s}}(\bm{r})\partial_{i}\Phi_{\bm{s}}(\bm{r})
\end{equation}
So at $\bm{s}=0$, $\int d\bm{r}\Phi_{\bm{s}}^{2}(\bm{r})\big|_{\bm{s}=0}=1$
and $\Phi_{\bm{s}}(\bm{r})\big|_{\bm{s}=0}=$ is real and therefor
\begin{align}
\partial_{i}\theta_{\bm{s}}\big|_{\bm{s}=0} & =\text{i}\int d\bm{r}\Phi_{\bm{s}}(\bm{r})\partial_{i}\Phi_{\bm{s}}(\bm{r})\big|_{\bm{s}=0}\nonumber \\
 & =\text{i}\braket{\Phi_{\bm{s}}\big|\partial_{i}\Phi_{\bm{s}}}\big|_{\bm{s}=0}\label{eq:ditheta}
\end{align}
According to Eqs.~(\ref{eq:diPsi2tild-r}, \ref{eq:ditheta}), we
know 
\begin{align}
\ket{\partial_{i}\tilde{\psi}_{2\bm{s}}}\big|_{\bm{s}=0} & =2\text{i}(\ket{\Phi_{\bm{s}}}\braket{\Phi_{\bm{s}}\big|\partial_{i}\Phi_{\bm{s}}}-\ket{\partial_{i}\Phi_{\bm{s}}})\big|_{\bm{s}=0}\nonumber \\
 & =-2\text{i}\ket{\partial_{i}^{0}\Phi_{\bm{s}}}\label{eq:diPsi2tild-sZero}
\end{align}
where $\ket{\partial_{i}^{0}\Phi_{\bm{s}}}$ is the projection of
$\ket{\partial_{i}\Phi_{\bm{s}}}$ onto the kernel of $\ket{\Phi_{\bm{s}}}\bra{\Phi_{\bm{s}}}$.
Therefore the satisfaction of Eq.~(\ref{eq:diPsi2tildPikalpha-reg})
is equivalent as 
\begin{equation}
\braket{\partial_{i}^{0}\Phi_{\bm{s}}\big|\pi_{k\alpha}}\big|_{\bm{s}=0}=0,\,\forall i,\,\alpha
\end{equation}

The saturation of the matrix bound associated with a null projector
requires that
\begin{equation}
\braket{\partial_{i}\tilde{\psi}_{1\bm{s}}\big|\pi_{k\alpha}}\big|_{\bm{s}=0}=\eta_{ij}^{k}\braket{\partial_{j}\tilde{\psi}_{1\bm{s}}\big|\pi_{k\alpha}}\big|_{\bm{s}=0},\,\forall i,\,j,\,\alpha,\label{eq:diPsi1tildPikalpha-null}
\end{equation}
\begin{equation}
\braket{\partial_{i}\tilde{\psi}_{2\bm{s}}\big|\pi_{k\alpha}}\big|_{\bm{s}=0}=\eta_{ij}^{k}\braket{\partial_{j}\tilde{\psi}_{2\bm{s}}\big|\pi_{k\alpha}}\big|_{\bm{s}=0},\,\forall i,\,j,\,\alpha.\label{eq:dipsi2tildePikalpha-null}
\end{equation}
Due to Eq.~(\ref{eq:diPsi-r-sZero}), Eq.~(\ref{eq:diPsi1tildPikalpha-null})
is trivially satisfied. Note that for null projectors $\braket{\Phi_{\bm{s}}\big|\pi_{k\alpha}}\big|_{\bm{s}=0}=0$
and therefore according to Eq.~(\ref{eq:diPsi2tild-sZero}), we find
\begin{equation}
\braket{\partial_{i}\tilde{\psi}_{2\bm{s}}\big|\pi_{k\alpha}}\big|_{\bm{s}=0}=2\text{i}\braket{\partial_{i}\Phi_{\bm{s}}\big|\pi_{k\alpha}}\big|_{\bm{s}=0}.
\end{equation}
Now the satisfaction of Eq.~(\ref{eq:diPsi2tildPikalpha-reg}) is
equivalent as is 
\begin{equation}
\braket{\partial_{i}\Phi_{\bm{s}}\big|\pi_{k\alpha}}\big|_{\bm{s}=0}=\eta_{ij}^{k}\braket{\partial_{j}\Phi_{\bm{s}}\big|\pi_{k\alpha}}\big|_{\bm{s}=0},\,\forall i,\,j,\,\alpha.
\end{equation}
\end{proof}

\subsubsection{Details of constructing the optimal measurement in the main text}

It is easily calculated that

\begin{equation}
\ket{\Phi_{\bm{s}}}\big|_{\bm{s}=0}=\ket{Z_{0}^{0}},\label{eq:Phizero}
\end{equation}

\begin{equation}
\ket{\partial_{1}\Phi_{\bm{s}}}\big|_{\bm{s}=0}=ik\ket{Z_{1}^{1}}/2,\label{eq:d1Phizero}
\end{equation}
\begin{eqnarray}
\ket{\partial_{2}\Phi_{\bm{s}}}\big|_{\bm{s}=0} & = & ik\ket{Z_{1}^{-1}}/2,\label{eq:d2Phi0}
\end{eqnarray}
\begin{equation}
\ket{\partial_{3}\Phi_{\bm{s}}}\big|_{\bm{s}=0}=-ik(\ket{Z_{2}^{0}}/3+\ket{Z_{0}^{0}})/2.\label{eq:d3Phizero}
\end{equation}
With Eqs. (\ref{eq:Phizero}-\ref{eq:d3Phizero}), one can easily
understand the details in the construction recipes in the main text.
For example, the following facts can be obtained:
\begin{equation}
\braket{\partial_{i}^{0}\Phi_{\bm{s}}\big|Z_{0}^{0}}\big|_{\bm{s}=0}=\braket{\partial_{i}\Phi_{\bm{s}}\big|Z_{0}^{0}}\big|_{\bm{s}=0}=0,\,i=1,\,2,
\end{equation}
\begin{align}
\braket{\partial_{3}^{0}\Phi_{\bm{s}}\big|Z_{0}^{0}}\big|_{\bm{s}=0} & =\braket{\partial_{3}\Phi_{\bm{s}}\big|Z_{0}^{0}}\big|_{\bm{s}=0}\nonumber \\
 & -\braket{\partial_{3}\Phi_{\bm{s}}\big|\Phi_{\bm{s}}}\big|_{\bm{s}=0}\braket{\Phi_{\bm{s}}\big|Z_{0}^{0}}\big|_{\bm{s}=0}\nonumber \\
 & =0.
\end{align}

\subsection{\label{subsec:sperp0}The case $\bm{s}_{\perp}=0$}

\subsubsection{Proof of Corollary 2 in the main text for the case of $\bm{s}_{\perp}=0$}

\begin{proof}We assume $\bm{s}_{\perp}=0$ and $s_{3}\neq0$, where
the rank of the state is strictly two. In this case, according to
Eqs.~(\ref{eq:vs}, \ref{eq:ws}), the explicit forms of $v_{\bm{s}}$
and $w_{\bm{s}}$ can be expressed as
\begin{equation}
v_{\bm{s}}\big|_{\bm{s}_{\perp}=0}=\frac{\pi\mathcal{A}^{2}}{ks_{3}}\sin(ks_{3}a^{2}),
\end{equation}
\begin{equation}
w_{\bm{s}}\big|_{\bm{s}_{\perp}=0}=-\frac{\pi\mathcal{A}^{2}}{ks_{3}}[1-\cos(ks_{3}a^{2})],
\end{equation}
\begin{equation}
\partial_{i}v_{\bm{s}}\big|_{\bm{s}_{\perp}=0}=\partial_{i}w_{\bm{s}}\big|_{\bm{s}_{\perp}=0}=0,\,i=1,\,2.
\end{equation}
According to Eqs. (\ref{eq:tan2theta}, \ref{eq:delta}), we obtain
\begin{equation}
\theta_{\bm{s}}\big|_{\bm{s}_{\perp}=0}=\frac{ks_{3}a^{2}}{4},\label{eq:theta-sperp0}
\end{equation}
\begin{equation}
\Delta_{\bm{s}}\big|_{\bm{s}_{\perp}=0}=\left[\frac{2}{ks_{3}a^{2}}\sin\left(\frac{ks_{3}a^{2}}{2}\right)\right]^{2},
\end{equation}

\begin{equation}
\frac{2\partial_{i}\theta_{\bm{s}}}{1+4\theta_{\bm{s}}^{2}}\Bigg|_{\bm{s}_{\perp}=0}=-\frac{\partial_{i}w_{\bm{s}}v_{\bm{s}}-w_{\bm{s}}\partial_{i}v_{\bm{s}}}{v_{\bm{s}}^{2}}\Bigg|_{\bm{s}_{\perp}=0},
\end{equation}
\begin{equation}
\partial_{i}\Delta{}_{\bm{s}}\big|_{\bm{s}_{\perp}=0}=\frac{v_{\bm{s}}\partial_{i}v_{\bm{s}}+w_{\bm{s}}\partial_{i}w_{\bm{s}}}{\Delta_{\bm{s}}}\Bigg|_{\bm{s}_{\perp}=0}.
\end{equation}
Therefore, we arrive at 
\begin{equation}
\partial_{i}\theta_{\bm{s}}\big|_{\bm{s}_{\perp}=0}=\partial_{i}\Delta{}_{\bm{s}}\big|_{\bm{s}_{\perp}=0}=0,\,i=1,\,2.\label{eq:ditheta-perpzero}
\end{equation}
According to Eqs. (\ref{eq:psi1r}, \ref{eq:psi2r}), we know
\begin{equation}
\ket{\partial_{i}\psi_{n\bm{s}}}=\frac{\ket{\partial_{i}\tilde{\psi}_{n\bm{s}}}}{\sqrt{4p_{n\bm{s}}}},\,n,\,i=1,\,2.
\end{equation}
Substituting Eqs.~(\ref{eq:diPsi1Psi2}, \ref{eq:diPsi2Psi1}, \ref{eq:ditheta-perpzero})
into Eqs.~(\ref{eq:Lemm1-cond1}, \ref{eq:Lemma1-cond2}), one obtains
the saturation condition for regular projectors. One can direct apply
Theorem 4 in the main text to obtain the saturation condition for
a null projector.

Furthermore, if near the critical point $\bm{s}=0$ the QFIM is saturated,
then by taking the limit $\bm{s}\to0$, it is also saturated at $\bm{s}=0$.\end{proof}

\subsubsection{Details of constructing the optimal measurement in the main text}

It can be calculated that according to Eqs. (\ref{eq:psiP-expression},
\ref{eq:psiM-expression}, \ref{eq:theta-sperp0}, \ref{eq:ditheta-perpzero}),
\begin{equation}
\tilde{\psi}_{1\bm{s}}(\bm{r})\big|_{\bm{s}_{\perp}=0}=2\mathcal{A}\text{circ}(r/a)\cos[ks_{3}(a^{2}-2r^{2})/4],
\end{equation}
\begin{equation}
\tilde{\psi}_{2\bm{s}}(\bm{r})\big|_{\bm{s}_{\perp}=0}=2\mathcal{A}\text{circ}(r/a)\sin[ks_{3}(a^{2}-2r^{2})/4],
\end{equation}
\begin{equation}
\partial_{i}\tilde{\psi}_{1\bm{s}}(\bm{r})\big|_{\bm{s}_{\perp}=0}=-kx_{i}\tilde{\psi}_{2\bm{s}}(\bm{r})\big|_{\bm{s}_{\perp}=0}\,i=1,\,2,
\end{equation}
\begin{equation}
\partial_{i}\tilde{\psi}_{2\bm{s}}(\bm{r})\big|_{\bm{s}_{\perp}=0}=kx_{i}\tilde{\psi}_{1\bm{s}}(\bm{r})\big|_{\bm{s}_{\perp}=0},\,i=1,\,2.
\end{equation}
We see that both $\tilde{\psi}_{1\bm{s}}(\bm{r})\big|_{\bm{s}_{\perp}=0}$
and $\tilde{\psi}_{2\bm{s}}(\bm{r})\big|_{\bm{s}_{\perp}=0}$ are
even while both $\partial_{i}\tilde{\psi}_{1\bm{s}}(\bm{r})\big|_{\bm{s}_{\perp}=0}$
and $\partial_{i}\tilde{\psi}_{2\bm{s}}(\bm{r})\big|_{\bm{s}_{\perp}=0}$
for $i=1,\,2$ are odd. Therefore, 
\begin{equation}
\braket{\tilde{\psi}_{1\bm{s}}\big|Z_{2n+1}^{2m+1}}\big|_{\bm{s}_{\perp}=0}=\braket{\tilde{\psi}_{2\bm{s}}\big|Z_{2n+1}^{2m+1}}\big|_{\bm{s}_{\perp}=0}=0,
\end{equation}
\begin{equation}
\braket{\partial_{i}\tilde{\psi}_{1\bm{s}}\big|Z_{2n}^{2m}}\big|_{\bm{s}_{\perp}=0}=\braket{\partial_{i}\tilde{\psi}_{2\bm{s}}\big|Z_{2n}^{2m}}\big|_{\bm{s}_{\perp}=0}=0,
\end{equation}
where $Z_{2n+1}^{2m+1}(\bm{r})$ is of odd parity and $Z_{2n}^{2m}(\bm{r})$
is of even parity. Furthermore, since
\begin{equation}
\partial_{1}\tilde{\psi}_{k\bm{s}}(\bm{r})\big|_{\bm{s}_{\perp}=0}\propto f(r)\cos\phi,\,k=1,\,2,
\end{equation}
\begin{equation}
\partial_{2}\tilde{\psi}_{k\bm{s}}(\bm{r})\big|_{\bm{s}_{\perp}=0}\propto f(r)\sin\phi,\,k=1,\,2,
\end{equation}
we obtain for $m\neq\pm1$
\begin{equation}
\braket{\partial_{1}\tilde{\psi}_{k\bm{s}}\big|Z_{2n+1}^{m}}\big|_{\bm{s}_{\perp}=0}=0,\,k=1,\,2,
\end{equation}
while 
\begin{equation}
\braket{\partial_{1}\tilde{\psi}_{k\bm{s}}\big|Z_{2n+1}^{\pm1}}\big|_{\bm{s}_{\perp}=0}\neq0,\,k=1,\,2.
\end{equation}

\subsection{\label{subsec:s30}The case $s_{3}=0$}

\subsubsection{The saturation conditions for the case of $s_{3}=0$}

Let us focus on the case where $s_{3}=0$ and $\bm{s}_{\perp}\neq0$.
According to Eqs.~(\ref{eq:vs}, \ref{eq:ws}), it is easily calculated
that
\begin{equation}
v_{\bm{s}}\big|_{s_{3}=0}=\mathcal{A}^{2}\int d\bm{r}\text{circ}(r/a)\cos(2k\bm{s}_{\perp}\cdot\bm{r}),
\end{equation}
\begin{equation}
w_{\bm{s}}\big|_{s_{3}=0}=0,
\end{equation}
and for $i=1,\,2$
\begin{equation}
\partial_{i}v_{\bm{s}}\big|_{s_{3}=0}=-2k\mathcal{A}^{2}\int d\bm{r}x_{i}\text{circ}(r/a)\sin(2k\bm{s}_{\perp}\cdot\bm{r}),
\end{equation}
\begin{equation}
\partial_{i}w_{\bm{s}}\big|_{s_{3}=0}=0.
\end{equation}
According to Eqs.~(\ref{eq:tan2theta}, \ref{eq:delta}), we obtain
\begin{equation}
\theta_{\bm{s}}\big|_{s_{3}=0}=0,\label{eq:thetas3zero}
\end{equation}
and for $i=1,\,2$

\begin{equation}
\partial_{i}\theta_{\bm{s}}\big|_{s_{3}=0}=0.\label{eq:dithetas3zero}
\end{equation}
Substituting Eq.~(\ref{eq:diPsi1Psi2}, \ref{eq:diPsi2Psi1}, \ref{eq:dithetas3zero})
into Eqs.~(\ref{eq:Lemm1-cond1}, \ref{eq:Lemma1-cond2}), one obtains
the following saturation condition for a regular projector $\Pi_{k}=\sum_{\alpha}\ket{\pi_{k\alpha}}\bra{\pi_{k\alpha}}$
where $\text{Tr}(\rho_{\bm{\lambda}}\Pi_{k})>0$ and the estimation
of the transverse separation $s_{1}$ and $s_{2}$,
\begin{equation}
\frac{\partial_{i}p_{1\bm{s}}}{p_{1\bm{s}}}\braket{\psi_{1\bm{s}}\big|\pi_{k\alpha}}+2\braket{\partial_{i}\psi_{1\bm{s}}\big|\pi_{k\alpha}}=\xi_{i}^{k}\braket{\psi_{1\bm{s}}\big|\pi_{k\alpha}},\label{eq:regular-cond1-s3zero}
\end{equation}
\begin{equation}
\frac{\partial_{_{i}}p_{2\bm{s}}}{p_{2\bm{s}}}\braket{\psi_{2\bm{s}}\big|\pi_{k\alpha}}+2\braket{\partial_{i}\psi_{2\bm{s}}\big|\pi_{k\alpha}}=\xi_{i}^{k}\braket{\psi_{2\bm{s}}\big|\pi_{k\alpha}}.\label{eq:regular-cond2-s3zero}
\end{equation}

One can direct apply Theorem 4 in the main text to obtain the following
saturation condition for a null projector $\Pi_{k}=\sum_{\alpha}\ket{\pi_{k\alpha}}\bra{\pi_{k\alpha}}$
where $\text{Tr}(\rho_{\bm{\lambda}}\Pi_{k})=0$,
\begin{equation}
\braket{\partial_{i}\tilde{\psi}_{n\bm{s}}\big|\pi_{k\alpha}}=\eta_{ij}\braket{\partial_{j}\tilde{\psi}_{n\bm{s}}\big|\pi_{k\alpha}},\,i,\,j=1,\,2,\forall n,\alpha.\label{eq:null-s3zero}
\end{equation}
Note that if near the critical point $\bm{s}=0$ the QFIM is saturated
by some optimal measurement, then by taking the limit $\bm{s}\to0$,
it is also saturated at $\bm{s}=0$.$\square$

\subsubsection{Details of constructing the optimal measurement in the main text}

It can be calculated according to Eqs.~(\ref{eq:psiP-expression},
\ref{eq:psiM-expression}, \ref{eq:thetas3zero}, \ref{eq:dithetas3zero})
that,
\begin{equation}
\tilde{\psi}_{1\bm{s}}(\bm{r})\big|_{s_{3}=0}=2\mathcal{A}\text{circ}(r/a)\cos(k\bm{s}_{\perp}\cdot\bm{r}),
\end{equation}
\begin{equation}
\tilde{\psi}_{2\bm{s}}(\bm{r})\big|_{s_{3}=0}=2\mathcal{A}\text{circ}(r/a)\sin(k\bm{s}_{\perp}\cdot\bm{r}),
\end{equation}
\begin{equation}
\partial_{i}\tilde{\psi}_{1\bm{s}}(\bm{r})\big|_{s_{3}=0}=-kx_{i}\tilde{\psi}_{2\bm{s}}(\bm{r})\big|_{s_{3}=0},\,i=1,\,2,
\end{equation}
\begin{equation}
\partial_{i}\tilde{\psi}_{2\bm{s}}(\bm{r})\big|_{s_{3}=0}=kx_{i}\tilde{\psi}_{1\bm{s}}(\bm{r})\big|_{s_{3}=0},\,i=1,\,2.
\end{equation}
We see that for $i=1,\,2$, both $\tilde{\psi}_{1\bm{s}}(\bm{r})\big|_{s_{3}=0}$
and $\partial_{i}\tilde{\psi}_{1\bm{s}}(\bm{r})\big|_{s_{3}=0}$ are
even while both $\tilde{\psi}_{2\bm{s}}(\bm{r})\big|_{s_{3}=0}$ and
$\partial_{i}\tilde{\psi}_{2\bm{s}}(\bm{r})\big|_{s_{3}=0}$ are odd.
We choose real basis with definite parity where the real even and
basis functions are denoted as $\pi_{\pm\alpha}(\bm{r})=\braket{\bm{r}\big|\pi_{\pm\alpha}}$
respectively. For an even and regular basis vector $\ket{\pi_{+\alpha}}$,
we can obtain $\braket{\psi_{2\bm{s}}\big|\pi_{k\alpha}}=0$ and $\braket{\partial_{i}\psi_{2\bm{s}}\big|\pi_{k\alpha}}=0$
by the parities of these functions. Thus both sides of Eq. (\ref{eq:regular-cond2-s3zero})
vanish and set no constraint on the constant $\xi_{i}^{k}$. From
Eq. (\ref{eq:regular-cond1-s3zero}) we find 
\begin{equation}
\xi_{i}^{+\alpha}=\left(\frac{\braket{\partial_{i}\psi_{1\bm{s}}\big|\pi_{+\alpha}}}{\braket{\psi_{1\bm{s}}\big|\pi_{+\alpha}}}+\frac{\partial_{i}p_{1\bm{s}}}{p_{1\bm{s}}}\right)\Bigg|_{s_{3}=0}
\end{equation}
is also real. For different regular even basis vectors, the coefficients
$\xi_{i}^{+\alpha}$ are not necessarily equal. Thus according to
the recipe in the main text, we obtain one regular projector $\Pi_{+\alpha}=\ket{\pi_{+\alpha}}\bra{\pi_{+\alpha}}$
corresponding to each of these vectors for the optimal measurement.
If an even basis vector $\ket{\pi_{+\alpha}}$ is null, then we see
that $\braket{\partial_{1}\tilde{\psi}_{2\bm{s}}\big|\pi_{+\alpha}}=\braket{\partial_{2}\tilde{\psi}_{2\bm{s}}\big|\pi_{+\alpha}}=0$
and
\begin{equation}
\eta_{21}^{+\alpha}=\frac{\braket{\partial_{2}\tilde{\psi}_{1\bm{s}}\big|\pi_{+\alpha}}}{\braket{\partial_{1}\tilde{\psi}_{1\bm{s}}\big|\pi_{+\alpha}}}\Bigg|_{s_{3}=0}
\end{equation}
is real. Again for different null even basis vectors, the coefficients
$\eta_{21}^{+\alpha}$ are not necessarily equal. We obtain one null
projector $\Pi_{+\alpha}=\ket{\pi_{+\alpha}}\bra{\pi_{+\alpha}}$
for each of these vectors for the optimal measurement. Similar analysis
can be done for odd basis functions, either regular or null. Therefore
one can construct the optimal projectors $\Pi_{-\alpha}=\ket{\pi_{-\alpha}}\bra{\pi_{-\alpha}}$.
So we conclude that rank one projectors formed by real and parity
definite basis vectors are optimal on the plane $s_{3}=0$.

\bibliographystyle{apsrev4-1}
\bibliography{Metrology}

\end{document}